\documentclass[showpacs,twocolumn,aps,prx,longbibliography,superscriptaddress,notitlepage]{revtex4-2}
\usepackage{qcircuit}
\usepackage[dvips]{graphicx}
\usepackage[normalem]{ulem}
\usepackage{amsmath,amssymb,amsthm,mathrsfs,amsfonts,dsfont}
\usepackage{subfigure, epsfig}
\usepackage{braket}
\usepackage{bm}
\usepackage{enumerate}
\usepackage{algorithm}
\usepackage{algpseudocode}
\usepackage{diagbox}
\usepackage{physics}
\usepackage{color}
\usepackage{multirow}
\usepackage[marginal]{footmisc}
\usepackage{comment}
\usepackage{tikz}
\usepackage{gensymb}
\usepackage[]{qcircuit}
\usetikzlibrary{arrows}
\usetikzlibrary{shapes,fadings,snakes}
\usetikzlibrary{decorations.pathmorphing,patterns}
\usetikzlibrary{calc}
\usetikzlibrary{positioning}
\usepackage[colorlinks = true]{hyperref}

\newtheorem{theorem}{Theorem}

\newtheorem{fact}{Fact}
\newtheorem{lemma}{Lemma}
\newtheorem{corollary}{Corollary}
\newtheorem{proposition}{Proposition}
\newtheorem{definition}{Definition}

\newcommand{\mc}{\mathcal}

\newcommand{\mbb}{\mathbb}

\newcommand{\comments}[1]{}

\begin{document}

\title{Separation between Entanglement Criteria and Entanglement Detection Protocols}

\begin{abstract}
Entanglement detection is one of the most fundamental tasks in quantum information science, playing vital roles in theoretical studies and quantum system benchmarking. 
Researchers have proposed many powerful entanglement criteria with high detection capabilities and small observable numbers. 
Nonetheless, entanglement criteria only represent mathematical rules deciding the existence of entanglement. 
The relationship between a good entanglement criterion and an effective experimental entanglement detection protocol (EDP) is poorly understood.
In this study, we introduce postulates for EDPs about their detection capabilities and robustness and use them to show the difference between entanglement criteria and EDPs.
Specifically, we design an entanglement detection task for unknown pure bipartite states and demonstrate that the sample complexity of any EDP and the number of observables for a good entanglement criterion can have exponential separation.
Furthermore, we discover that the optimal EDP with the lowest sample complexity does not necessarily correspond to the optimal entanglement criterion with the fewest observables.
Our results can be used to prove the exponential speedups achieved through quantum memory and be generalized to multipartite entanglement detection.
By highlighting the significance and independence of EDP design, our work holds practical implications for entanglement detection experiments.

\end{abstract}

\author{Zhenhuan Liu}
\thanks{These authors contributed equally to this work.}
\affiliation{Center for Quantum Information, Institute for Interdisciplinary Information Sciences, Tsinghua University, Beijing 100084, China}

\author{Fuchuan Wei}
\thanks{These authors contributed equally to this work.}
\affiliation{Yau Mathematical Sciences Center and Department of Mathematics, Tsinghua University, Beijing 100084, China}

\maketitle

\textbf{Introduction.}
The ``spooky action", entanglement, has long been recognized as a crucial resource that allows quantum physics to exhibit advantages over classical systems in tasks such as quantum communication and simulation \cite{Horodecki2009entanglement}.
Consequently, the task of entanglement detection holds practical and theoretical significance, leading to the proposal of numerous approaches \cite{GUHNE2009detection}.
However, it is important to note that most of these proposals are essentially entanglement criteria (EC), which are mathematical rules that provide sufficient conditions for the presence of entanglement.
Experimentally implementing an EC requires additional effort to design a practical entanglement detection protocol (EDP).

The positive map criterion \cite{GUHNE2009detection} is a typical example to show the gap between EC and EDP, which detects entanglement by telling the positivity of the density matrix after a positive but not completely positive map.
Including the positive partial transposition (PPT) criterion \cite{peres1996ppt}, the positive map criterion finds widespread use in theoretical studies of entanglement due to its concise mathematical form and strong detection capability \cite{vidal2002computable,calabrese2012field,shapourian2021negativity}.
Two main approaches have been developed to employ the positive map criterion in real experiments.
The first approach is measuring entanglement witnesses derived from the corresponding positive map criterion \cite{GUHNE2009detection}, such as $W=\ketbra{\psi}{\psi}^{\mathrm{T}_A}$ for PPT criterion, where $\mathrm{T}_A$ represents the partial transposition operation on subsystem $A$.
When $\ket{\psi}$ has a simple form, this method is experimentally friendly while exponentially weakening the detection capability of the original positive map criterion \cite{liu2022fundamental}. 
The second approach entails measuring the moments of the density matrices processed by the positive map and using these moments to infer the positivity of the target matrix \cite{gray2018ml,elben2020mixed,zhou2020Single,yu2021optimal,Neven2021resolved}. 
This method has much higher detection capability compared with the entanglement witness approach \cite{liu2022fundamental,yu2021optimal}, while moment measurement requires exponential experimental times \cite{elben2020mixed,Neven2021resolved,zhou2020Single,liu2022permutation}.
In conclusion, although positive map criteria are theoretically powerful, applying them in real experiments requires compromising the detection capability or dealing with high experimental complexity.

Over the years, researchers gradually realized that there might exist some fundamental limitations in entanglement detection \cite{znidaric2007EW,bhosale2012wigner,aubrun2014entanglement,collins2016random}, especially when they tried to construct a powerful EC with a minimal number of observables \cite{weilenmann2020faithful,gunhe2021faithful,zhou2019decomposition,Bae2022MUB,morelli2023resource}.
Remarkably, the fundamental limitation of EC was derived in Ref.~\cite{liu2022fundamental}, where the authors established a fundamental trade-off between the observable number and the detection capability of EC.
However, this result cannot be directly applied to analyze the performances of EDPs as the number of observables does not directly equal the experimental sample complexity.

With the rapid development of quantum hardware, entanglement detection gradually transforms from a theoretical problem to an essential task for quantum simulation and benchmarking \cite{Haffner2005ion,Leibfried2005cat,monz2011fourteen,pan2012multiphoton,adam2016thermalization,omran2019cat,brydges2019probing,chao2019twenty,wang2018sixphoton,Cao2023super}.
At the same time, performances of entanglement detection protocols in practical scenarios attract more and more attention \cite{zhou2020coherent,morelli2022imprecise,Miller2023graph}.
To provide a guide for designing entanglement detection experiments, it is an urgent problem to figure out the complexities and limitations of EDPs.
Specifically, some critical questions arise: does a good EC correspond to a good EDP? Does a fundamental limitation of EDP also exist? Can we use some quantum resources to reduce the sample complexity of EDP?

In this work, we provide answers to these questions.
We first clarify the difference between EC and EDP and propose two postulates of a good EDP about its detection capability and error robustness.
Then, we design an entanglement detection task that can be solved efficiently using an EC constructed with a single observable, while any EDP satisfying our postulates requires exponential repetition times. 
This example demonstrates that the observable number of EC and sample complexity of EDP can have exponential separation.
Furthermore, in this task, we find that the optimal EDP with the lowest sample complexity does not correspond to the optimal EC with minimal observables, further showing that designing EDP and EC are two independent tasks.
In addition, our results help to show that quantum memory can bring exponential speedups for EDP.
We hope our work can remind entanglement detection experimenters that the practical performance of the corresponding EDP is a nontrivial problem that should be considered independently from EC when designing experiments.

\textbf{EC and EDP.}
Entangled states are those that cannot be written in a separable form
\begin{equation}
\rho=\sum_{i}p_i\rho_{A}^i\otimes\rho_{B}^i,
\end{equation}
where $p_i\ge 0$ represents probability, $A$ and $B$ represent two subsystems. Entanglement criteria are theoretical rules that distinguish some entangled states from separable ones \cite{GUHNE2009detection}. Mathematically speaking, an EC can be represented as a function of the target state $C(\rho)$, which satisfies that $C(\rho)\ge 0$ for all the separable $\rho$. 
For example, $C(\rho)=\Tr(W\rho)$ represents the entanglement witness criterion and $C(\rho)=\lambda_{\mathrm{min}}(\rho^{\mathrm{T}_A})$ represents the PPT criterion, where $W$ is a witness operator and $\lambda_{\mathrm{min}}$ represents the minimal eigenvalue.
As shown in Fig.~\ref{fig:EC_EDP}(a), given a target state, a valid EC tells you if the state is entangled when $C(\rho)<0$, or if it is uncertain when $C(\rho)\ge 0$. 
As EC are just theoretical rules, we can ask EC to never make mistakes by classifying separable states into entangled ones. 
Thus, to design a powerful EC, we only need to optimize the ratio of entangled states that can be detected by it, i.e., the detection capability.

\begin{figure}
\centering
\includegraphics[width=0.48\textwidth]{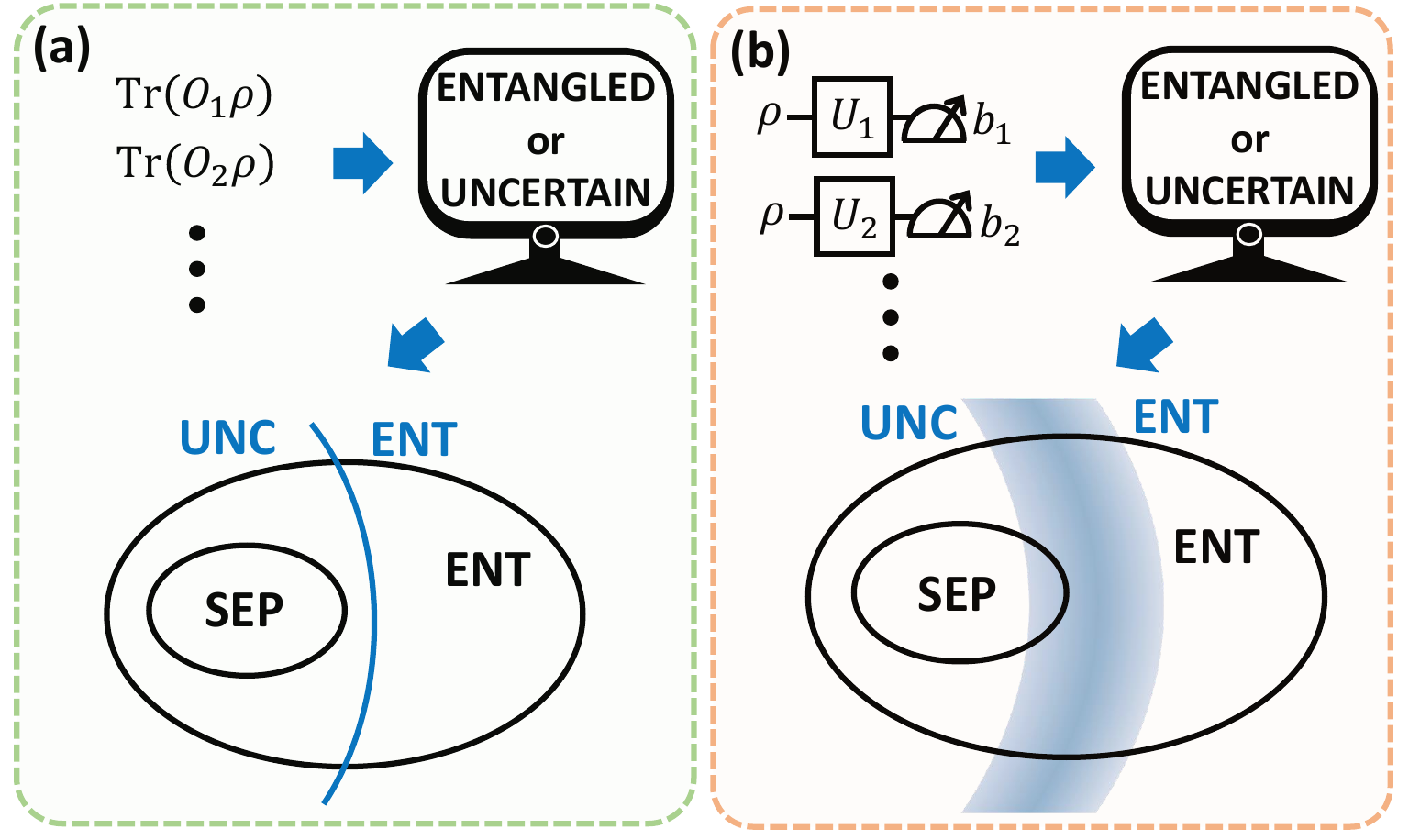}
\caption{
\emph{Difference between EC and EDP.}
ENT and SEP represent sets of entangled and separable states, respectively.
Here, we only consider EC and EDPs constructed by single-copy observable expectation values and quantum experiments.
(a) EC represents mathematical rules to decide whether a state is entangled and never makes mistakes.
The blue line represents those states with $C(\rho)=0$, states at the left of the line satisfy $C(\rho)>0$, and vice versa.
(b) EDP represents rules of physical operations and measurements to infer entanglement with finite rounds of experiments. 
The blue line for $C(\rho)$ becomes a blue area for $\hat{C}(\rho)$ whose darkness is proportional to the probability for a state to satisfy $\hat{C}(\rho)=0$.
In principle, any state has a nonzero probability of being classified as entangled or uncertain (UNC) by an EDP.
}
\label{fig:EC_EDP}
\end{figure}

In real experiments, we deal with EDP, which represents rules about operating and measuring the target quantum states and using the measurement results to infer whether the target state is entangled. 
EDP also decides the entanglement of the target state through a function, $\hat{C}(\rho)$, whose value is calculated using the experiment data. 
Due to the intrinsic randomness in quantum measurements and the finite experiment times, the value of $\hat{C}(\rho)$ has uncertainty in different rounds of experiments, as shown in Fig.~\ref{fig:EC_EDP}(b). 
Thus, there always exists the possibility for EDP to classify a separable state into an entangled one and make mistakes, which is the most important difference between EC and EDP. 
We can thus summarize two postulates for a good EDP about its detection capability and robustness to errors.
\begin{enumerate}
\item Completeness: the probability of outputting the result of entanglement should be large enough, i.e., 
\begin{equation}\label{eq:completeness}
\mathrm{Pr}[\hat{C}(\rho)<0]=\Theta(1);
\end{equation}
\item Soundness: the probability of actually being entangled when outputting the result of entanglement should be large enough, i.e., 
\begin{equation}\label{eq:soundness}
\mathrm{Pr}[\rho\in\mathrm{ENT}|\hat{C}(\rho)<0]=\frac{1}{2}+\Theta(1).
\end{equation}
\end{enumerate}
Note that completeness is the only requirement for EC.
As an EDP is commonly constructed based on some EC, the first postulate requires that the EDP should be based on a good EC, while the second one mainly puts the requirement for the sample complexity to suppress the error.

To analyze the performances of EC and EDP, we need to specify the underlying distribution of quantum states.
In practical situations, the state distribution is determined by experiment setups and can be very concentrated.
Meanwhile, the performance of entanglement detection when lacking prior knowledge of the target state is also an important indicator, which reflects the general detection power of entanglement detection methods. 
A parameterized distribution $\pi_{d,k}$, which is induced by a higher-dimensional random pure-state distribution \cite{Karol2001induce}, is commonly adopted to analyze performances of different entanglement detection methods \cite{collins2016random}.

\begin{definition}
A state $\rho\in D(\mc{H}_d)$ is said to be sampled according to the distribution $\pi_{d,k}$ if and only if it is the reduced density matrix of a $d\times k$-dimensional random pure state $\ketbra{\psi}{\psi}\in D(\mc{H}_{d\times k})$. 
Here, by random, we mean that $\ket{\psi}=U\ket{0}$ where $U$ is a Haar-measure random unitary matrix and $\ket{0}$ is a fixed pure state.
\end{definition}

Without loss of generality, hereafter, we consider $\rho$ has two subsystems $A$ and $B$ with dimensions $d_A=d_B=\sqrt{d}$. 
Based on this state distribution, it is proved that any EC constructed by expectation values measured on single copy of $\rho$ has a fundamental limitation \cite{liu2022fundamental}.

\begin{fact}\label{fact:fundamental}
Consider an entanglement criterion $C(\rho)$ which only depends on the expectation values of single-copy observables, $C(\rho)=f(\Tr(O_1\rho),\cdots,\Tr(O_M\rho))$ with $f(\cdot)$ being a multivariate function. If we want to make sure the probability $\mathrm{Pr}_{\rho\sim\pi_{d,k}}\left[C(\rho)<0\right]$ is a positive constant, the number of observables should at least $M=\tilde{\Omega}(k)$.
\end{fact}

For EC, the state distribution is only employed to calculate the detection power, i.e., the completeness as defined earlier. 
For EDP, the state distribution is also adopted to specify the robustness to errors, i.e., the soundness.

\textbf{Exponential Separation.}
In this section, we will use a specific entanglement detection task to show the gap between EC and EDP.
The distribution $\pi_{d,k}$ is not a good example, as it has been proved that a state sampled according to $\pi_{d,k}$ is entangled with probability asymptotically $1$ for $k$ in a large range \cite{aubrun2012phase}.
Thus, if directly considering $\pi_{d,k}$, there exists a simple EDP that is just outputting ``entangled," as EDP allows for a nonzero error probability.
To benefit our following discussions, we amplify the separable component in $\pi_{d,k}$ and define a new state distribution $\pi_{d,k}^*$.

\begin{definition}
A state $\rho\in D(\mathcal{H}_d)$ is said to be sampled according to the distribution of $\pi_{d,k}^*$ if and only if with probability $0.5$, it is sampled from $\pi_{d,k}$; with probability $0.5$, it is the tensor product of two random states $\rho_A\otimes\rho_B$, which are independently sampled from $\pi_{\sqrt{d},\sqrt{k}}$.
\end{definition}

In this section, we mainly consider the case of pure states, i.e., $k=1$, while leaving the discussion of larger values of $k$ in Appendix \ref{sec:generalization} .
As the only requirement for EC is completeness, we only care about the number of observables needed for detecting the entanglement of $\pi_{d,1}^*$ with a constant probability.
As $\pi_{d,1}^*$ contains a constant portion of $\pi_{d,1}$, Fact~\ref{fact:fundamental} shows that a small number of observables might be enough to detect the entanglement of $\pi_{d,1}^*$ with a constant probability, no matter how large $d$ is.
Fortunately, this is indeed the case.
We use the SWAP operator $S$ defined between subsystems $A$ and $B$, which satisfies $S\ket{\psi_A,\psi_B}=\ket{\psi_B,\psi_A}$, to show this.
The SWAP operator is actually an entanglement witness as $S=\ketbra{\Phi^+}{\Phi^+}^{\mathrm{T}_A}$ with $\ket{\Phi^+}=\sum_{i=1}^{\sqrt{d}}\ket{ii}$ being the unnormalized maximally entangled state. 
\begin{theorem}\label{theorem:Ratio}
Considering that $\rho$ is sampled according to $\pi_{d,1}^*$, the detection power of $S$ satisfies
\begin{equation}
\underset{\rho\sim\pi_{d,1}^*}{\mathrm{Pr}}\left[\Tr(S\rho)<0\right]=\frac{1}{2}I_{1/2}\left(\frac{d+\sqrt{d}}{2},\frac{d-\sqrt{d}}{2}\right),
\end{equation}
where $I_{\cdot}(\cdot,\cdot)$ is the regularized incomplete beta function.
\end{theorem}
We leave the detailed proof in Appendix~\ref{sec:SWAP_analysis} and conjecture that the detection power of $S$ is asymptotically a constant, $\lim_{d\to \infty}I_{1/2}(\frac{d+\sqrt{d}}{2},\frac{d-\sqrt{d}}{2})=\mathrm{Const}$.
We numerically test this conjecture using Fig.~\ref{fig:PPT_witness}, in which the detection power converges to some constant when the dimension increases.

\begin{figure}
\centering
\includegraphics[width=0.45\textwidth]{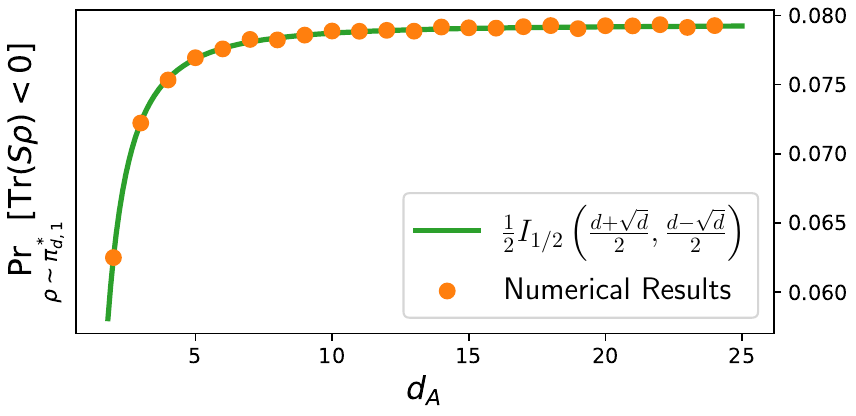}
\caption{The probability for detecting the entanglement in $\pi_{d,1}^*$ using the entanglement witness $S=\ketbra{\Phi^+}{\Phi^+}^{\mathrm{T}_A}$. The $x$-axi represents the dimension of the subsystem $A$, which is the square root of the dimension of the whole system. For each orange point, we sample $10^7$ $\rho$'s from $\pi_{d,1}^*$ and calculate the ratio of these states satisfying $\Tr(S\rho)<0$.}
\label{fig:PPT_witness}
\end{figure}

However, considering the physical implementation of this EC, i.e., the corresponding EDP, which also depends on single-copy quantum operations, we find that the sample complexity is high.
The first reason is that if we randomly select a state according to $\pi_{d,1}^*$, the expectation value is normally exponentially small, around $\frac{1}{\sqrt{d}}$.
Thus, one should keep the measurement error exponentially small to estimate a negative value with high certainty.
Secondly, considering practical situations, when measuring $S$, in every experiment time, one randomly gets a result $+1$ or $-1$, which are eigenvalues of $S$.
Thus, the variance of a single experiment result is a constant.
According to Hoeffding's inequality, to hold the soundness requirement for EDP, the complexity is around $\mathcal{O}(d)$.

From this example, we find that if we directly design an EDP based on a good EC, the sample complexity might be exponentially high. 
One may further ask if a better EDP for this entanglement detection task exists or if a fundamental limitation exists.
In the following, we will show that a better EDP does exist, and at the same time, \emph{an exponential sample complexity lower bound also exists for any EDP, although this is an easy task for EC in terms of the observable number}.

\begin{theorem}\label{theorem:exponential_complexity}
Given multiple copies of a state $\rho$ that is sampled from $\pi_{d,1}^*$, if only single-copy quantum operations are allowed, the minimum number of copies required by an adaptive EDP with completeness, $\Pr_{\rho\sim\pi_{d,1}^*}\left[\hat{C}(\rho)<0\right]=\frac{1}{4}$, and soundness, $\Pr_{\rho\sim\pi_{d,1}^*}\left[\rho\in\operatorname{ENT}\Big|\hat{C}(\rho)<0\right]=\frac{5}{6}$, scales as $\Theta(d^{1/4})$.
\end{theorem}

We leave the detailed proof in the Appendix~\ref{sec:bipartite_lower_bound} and sketch core ideas here. 
To prove the lower bound, we first notice that if the EDP described in this theorem exists, it can be used to solve a state discrimination task.
Specifically, given multiple copies of a state $\rho$, the state discrimination task is to decide whether it is sampled from the set $s_1 = \{\rho\sim\pi_{d,1}\}$ or $s_2=\{\rho=\rho_A\otimes\rho_B,\rho_A,\rho_B\sim\pi_{\sqrt{d},1}\}$, with equal probabilities.
With the EDP given in Theorem~\ref{theorem:exponential_complexity}, one can classify the state to $s_1$ when the EDP outputs $\hat{C}(\rho)<0$ and to $s_2$ otherwise.
With the conditions of completeness and soundness, the success probability of this strategy reaches $\frac{2}{3}$.
Therefore, if this state discrimination task has a sample complexity lower bound, it is also a lower bound of the EDP.

\begin{figure}
\centering
\includegraphics[width=0.35\textwidth]{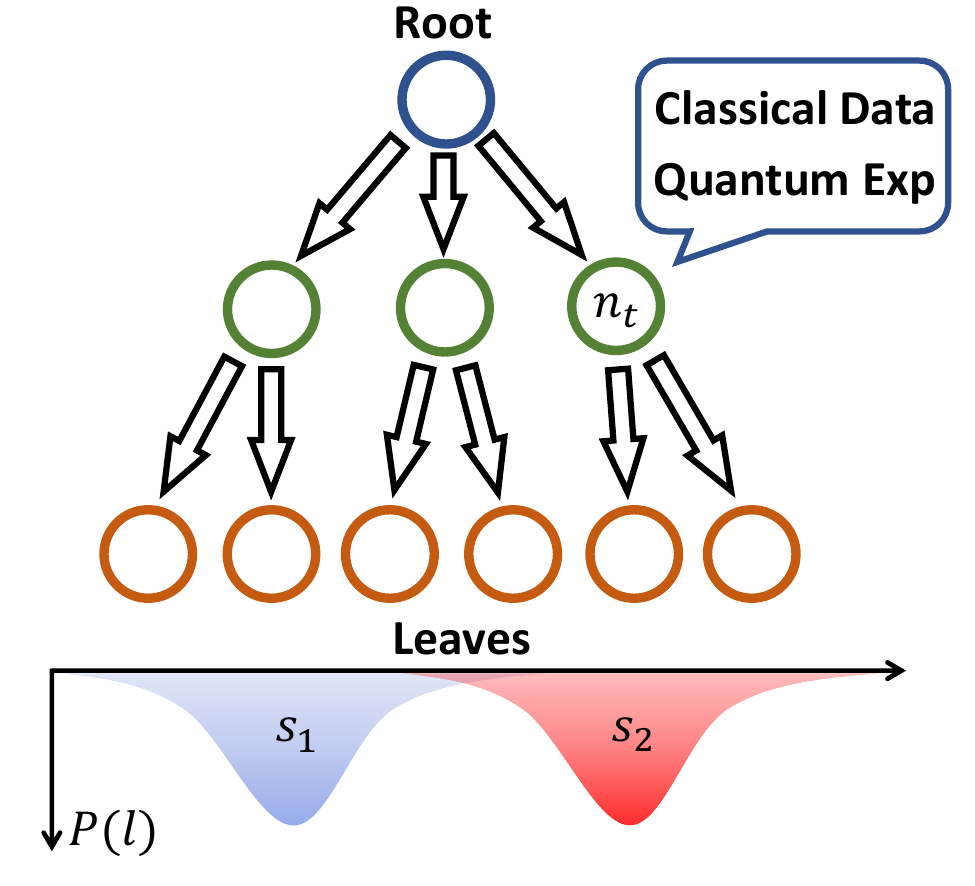}
\caption{Every state discrimination protocol can be represented as a decision tree, where $n$ and $l$ stand for the node and leaf, respectively.
$P(l)$ stands for the probability distribution on leaves, which is determined by the protocol and input state.
The $t$-th layer of the decision tree represents the $t$-th step of the protocol. 
In our setting, the experimenter has a classical memory that records the experiment data of the former steps and a quantum processor that can perform operations and measurements on quantum states. 
The protocol has a rule that decides how to perform quantum experiments based on the data recorded in the classical memory.  
At the end of the experiment, the experimenter will decide which set the target state comes from, depending on the leave the experimenter reaches. }
\label{fig:tree}
\end{figure}

To find the lower bound for this state discrimination task, we employ a decision tree to describe all the methods, including those depending on adaptive strategies \cite{Huang2022quantum,Chen2021exponential}, as shown in Fig.~\ref{fig:tree}.
Every node of the tree has all the experiment data of its child nodes and uses it to decide how to experiment within this step.
After many steps, different input states will lead to different probability distributions on the leaves.
If the deviation between two probability distributions on the leaves is large enough, we can discriminate two state sets with high probability.
We prove that, when only single-copy quantum operations are allowed in this decision tree, the depth of this tree should be at least $\Omega(d^{1/4})$ to make the distance of probability distributions caused by $s_1$ and $s_2$ a constant, which is also the sample complexity lower bound for single-copy EDP.

To prove the upper bound, we consider an EDP that satisfies the requirements of Theorem \ref{theorem:exponential_complexity} and derive its sample complexity.
This EDP is based on the purity criterion \cite{GUHNE2009detection}, which detects entanglement by measuring $C(\rho)=\Tr(\rho_{A}^2)-\Tr(\rho^2)$.
As states in $\pi_{d,1}^*$ are all pure states, the condition of $\Tr(\rho_{A}^2)<1$ can help us to detect the entanglement of $\rho$.
According to the properties of random pure states, if a state is sampled according to $\pi_{d,1}^*$, $\Tr(\rho_A^2)$ is either one or exponentially small \cite{page1993curve}, with probability $0.5$ for each scenario.
Thus, any protocol that accurately estimates $\Tr(\rho_A^2)$ at least satisfies the requirement for completeness, Eq.~\eqref{eq:completeness}.
Then, we need to find a purity measurement protocol that satisfies soundness with the lowest sample complexity.
Due to the restriction of single-copy quantum operations, we employ the randomized measurement protocol \cite{elben2019toolbox,Elben2023toolbox}, which is commonly employed for entanglement detection \cite{cieslinski2023analysing,imai2023randomized}.
In this protocol, one evolves the subsystem $A$ using many random unitaries, measures it in the computational basis, and computes the purity from these experiment data.
In Appendix~\ref{sec:upper_bound_proof}, we analytically show that $\mathcal{O}(d^{1/4})$ measurement times is enough to suppress the variance of the randomized measurement protocol to satisfy the soundness requirement, Eq.~\eqref{eq:soundness}, which concludes our proof for the upper bound.

There are some important consequences from Theorem \ref{theorem:exponential_complexity} and its proof.
Firstly, \emph{the optimal EDP with the lowest sample complexity does not necessarily correspond to the optimal EC with the smallest number of observables.}
As stated before, the optimal single-copy EDP is based on the randomized measurement protocol, which requires measurements on many bases.
If we judge entanglement detection protocols only regarding the number of observables, this randomized measurement EDP is much worse than the entanglement witness EDP we mentioned before, which only requires a single observable.
Secondly, a counter-intuitive fact from Theorem \ref{theorem:exponential_complexity} is that entanglement detection for pure states can be exponentially hard when considering sample complexity. 
However, such a task is well-acknowledged to be easy for EC and theoretical analysis.
Note that if considering practical situations with enough prior knowledge of the target state, we can largely avoid the exponential sample complexity.

In addition, Ref.~\cite{liu2022fundamental} has shown that quantum memory, or equivalently joint operations, brings advantages in constructing EC by exponentially saving the number of observables.
Using conclusions in this work, we can further conclude that quantum memory also brings exponential speedups for practical entanglement detection tasks.
This is because the SWAP test algorithm, which relies on quantum memory, estimates purity with constant sample complexity. 
We thus arrive at the following conclusion, which is also proved in Ref.~\cite{wei2023realizing} with another quantum algorithm.
\begin{corollary}
Quantum memory can bring exponential speedups in sample complexities for entanglement detection protocols.  
\end{corollary}

\textbf{Discussion.}
We summarize the take-home messages of this work.
Firstly, in addition to the detection power, the only requirement for an EC, an effective EDP must also possess robustness against errors.
Secondly, certain entanglement detection tasks exist that are relatively easy for an EC, achieving constant detection probability with a small number of observables while proving to be exponentially challenging for any EDP in terms of sample complexity.
Thirdly, although an EDP is typically designed based on an EC, the optimal EDP may not necessarily correspond to the optimal EC, as their evaluation criteria differ.
Lastly, using quantum memory can significantly reduce the sample complexity of EDP, demonstrating the great potential of multi-copy entanglement detection protocols \cite{rico2024poly}.

In this study, we primarily focus on pure bipartite states. 
Our results, especially the exponential lower bounds, can be extended to mixed-state and multipartite scenarios, as discussed in Appendix~\ref{sec:generalization}. Specifically, for $\pi_{d,k}^*$ with $k>1$, we enhance the lower bound of sample complexity from $\Omega(d^{1/4})$ to $\Omega((dk)^{1/4})$. 
For multipartite entanglement detection, we introduce a new state distribution and establish a tight exponential lower bound on the complexity of multipartite EDP.

While our results are derived based on a specific state distribution, they represent an initial step towards systematically considering sample complexities of entanglement detection protocols. 
We hope that our work will inspire further exploration using various practical state distributions, ultimately reducing the cost of real-world entanglement detection experiments. 
Additionally, our main result for EDP complexities is derived under the restriction of single-copy operations. 
It would also be valuable to consider other types of operations, such as separable and local operations \cite{Harrow2023orthogonality}.

\textbf{Acknowledgement.}
The authors would like to thank Satoya Imai, Xiongfeng Ma, Otfried Gühne, Senrui Chen, Pei Zeng, Rui Zhang, Zihao Li, Pengyu Liu, Boyang Chen, Xueying Yang, and Rundi Lu for their insightful discussions and suggestions.
We express special thanks to Daniel Miller for his vital and generous help during the final stages of this work.
ZL acknowledges support from the National Natural Science Foundation of China Grant No. 12174216 and the Innovation Program for Quantum Science and Technology Grant No. 2021ZD0300804. 
FW is supported by BMSTC and ACZSP Grant No.~Z221100002722017.


%

\appendix

\onecolumngrid
\newpage
\setcounter{theorem}{0}
\setcounter{definition}{0}
\setcounter{proposition}{0}
\setcounter{corollary}{0}

\section{Notations and Auxiliary Lemmas}

We use $\mc{H}_d$ to represent the $d$-dimensional Hilbert space. $D(\mc{H}_d)$ represents the set of density matrices on $\mc{H}_d$. 
For convenience, sometimes we omit the Dirac notation and denote $\ketbra{\psi}{\psi}$ as $\psi$. 

Let $S_T$ be the symmetric group of order $T$. 
For distinct $a_1,\cdots,a_l$, we use the notation $(a_1\cdots a_l)$ to denote a cyclic permutation in $S_T$, which acts as $a_1\mapsto a_2$, $\cdots$, $a_{l-1}\mapsto a_l$, $a_l\mapsto a_1$. 
For $\sigma\in S_T$, define a unitary operator $W_\sigma$ acting on $\mc{H}_d^{\otimes T}$ by
\begin{equation}
W_\sigma=\sum_{i_1,\cdots, i_T=1}^{d}\ketbra{i_{\sigma^{-1}(1)}\cdots i_{\sigma^{-1}(T)}}{i_1\cdots i_T},
\end{equation}
where $\{\ket{i}\}_{i=1}^d$ is an orthonormal basis of $\mc{H}_d$.

For two discrete probability distributions $p=\{p_1,\cdots,p_N\}$ and $q=\{q_1,\cdots,q_N\}$, define the total variation between them by $\operatorname{TV}(p,q)=\frac{1}{2}\sum_{i=1}^N\abs{p_i-q_i}$.

For two systems $\mc{H}_1$ and $\mc{H}_2$, define the unravelling operation \cite{wood2011tensor} $\mathcal{V}_M:\mc{H}_1^{\otimes M}\otimes\mc{H}_2^{\otimes M}\rightarrow\left(\mc{H}_1\otimes\mc{H}_2\right)^{\otimes M}$ by 
\begin{equation}
\mc{V}_M(\ket{x_1}\otimes\cdots\otimes\ket{x_M})\otimes(\ket{y_1}\otimes\cdots\otimes\ket{y_M})=(\ket{x_1}\otimes\ket{y_1})\otimes\cdots\otimes(\ket{x_M}\otimes\ket{y_M}),
\end{equation}
which holds for all $\ket{x_i}\in\mc{H}_1$ and $\ket{y_i}\in\mc{H}_2$. More generally, for $K$ systems $\mc{H}_1,\cdots,\mc{H}_K$, define $\mathcal{V}_M^{(K)}:\mc{H}_1^{\otimes M}\otimes\cdots\otimes\mc{H}_K^{\otimes M}\rightarrow\left(\mc{H}_1\otimes\cdots\otimes\mc{H}_K\right)^{\otimes M}$ by 
\begin{equation}
\begin{aligned}
&\mc{V}_M^{(K)}\left[\left(\ket{x_1^{(1)}}\otimes\cdots\otimes\ket{x_M^{(1)}}\right)\otimes\cdots\otimes\left(\ket{x_1^{(K)}}\otimes\cdots\otimes\ket{x_M^{(K)}}\right)\right]\\
=&\left(\ket{x_1^{(1)}}\otimes\cdots\otimes\ket{x_1^{(K)}}\right)\otimes\cdots\otimes\left(\ket{x_M^{(1)}}\otimes\cdots\otimes\ket{x_M^{(K)}}\right),
\end{aligned}
\end{equation}
which holds for all $\ket{x_i^{(k)}}\in\mc{H}_k$, where $k=1,\cdots,K$ and $i=1,\cdots,M$.

\begin{lemma}[\cite{harrow2013church}]\label{lemma:expectation_to_permutation}
For Haar random pure state $\ket{\psi}\in\mc{H}_d$, we have
\begin{equation}
\underset{\psi\sim \operatorname{Haar}}{\mathbb{E}}\ketbra{\psi}{\psi}^{\otimes T}=\frac{1}{\left(d+T-1\right)\cdots\left(d+1\right)d}\sum_{\sigma\in S_{T}}W_{\sigma}.
\end{equation}
\end{lemma}

\section{Detection Power and Performance of SWAP Operator}\label{sec:SWAP_analysis}
The entanglement witness $S$ is an important example in the main context. 
In this section, we will first evaluate its detection power as an EC and then estimate its sample complexity as an EDP.

\begin{theorem}
Consider that $\rho$ is sampled according to $\pi_{d,1}^*$, the detection power of $S$ satisfies
\begin{equation}
\underset{\rho\sim\pi_{d,1}^*}{\mathrm{Pr}}\left[\Tr(S\rho)<0\right]=\frac{1}{2}I_{1/2}\left(\frac{d+\sqrt{d}}{2},\frac{d-\sqrt{d}}{2}\right),
\end{equation}
where $I_{\cdot}(\cdot,\cdot)$ is the regularized incomplete beta function.
\end{theorem}
\begin{proof}
We take an orthonormal basis $\left\{\big|\alpha_1\big\rangle,\cdots,\big|\alpha_{\frac{d+\sqrt{d}}{2}}\big\rangle\right\}$ for the $+1$ eigenspace of $S$, and an orthonormal basis $\left\{\big|\alpha_{\frac{d+\sqrt{d}}{2}+1}\big\rangle,\cdots,\big|\alpha_{d}\big\rangle\right\}$ for the $-1$ eigenspace of $S$.
For a random complex Gaussian vector $v\sim\mc{N}_{\mbb{C}}(0,\mbb{I}_d)$, take $\ket{\psi}=\sum_{x=1}^d\frac{v_x}{\norm{v}}\ket{\alpha_x}\in\mc{H}_d$, then $\psi\sim\pi_{d,1}$ \cite{Collins2015Random}. 
Note that $\Tr(S\psi)=\sum_{x=1}^{\frac{d+\sqrt{d}}{2}}\frac{\abs{v_x}^2}{\norm{v}^2}-\sum_{\frac{d+\sqrt{d}}{2}+1}^{d}\frac{\abs{v_x}^2}{\norm{v}^2}$. We have
\begin{equation}
\begin{aligned}
\underset{\rho\sim\pi_{d,1}^*}{\mathrm{Pr}}\left[\Tr(S\rho)<0\right]=&\frac{1}{2}\underset{\psi\sim\pi_{d,1}}{\mathrm{Pr}}\left[\Tr(S\psi)<0\right]+\frac{1}{2}\underset{\psi_A,\psi_B\sim\pi_{\sqrt{d},1}}{\mathrm{Pr}}\left[\Tr(S\psi_A\otimes\psi_B)<0\right]\\
=&\frac{1}{2}\underset{\psi\sim\pi_{d,1}}{\mathrm{Pr}}\left[\Tr(S\psi)<0\right]\\
=&\frac{1}{2}\underset{v\sim\mc{N}_{\mbb{C}}(0,\mbb{I}_d)}{\mathrm{Pr}}\left[\sum_{x=1}^{\frac{d+\sqrt{d}}{2}}\abs{v_x}^2\Bigg/\sum_{\frac{d+\sqrt{d}}{2}+1}^{d}\abs{v_x}^2<1\right].
\end{aligned}
\end{equation}
For $v\sim\mc{N}_{\mbb{C}}(0,\mbb{I}_d)$, we have $\sum_{x=1}^{\frac{d+\sqrt{d}}{2}}\abs{v_x}^2\sim\Gamma\left(\frac{d+\sqrt{d}}{2},1\right)$ and $\sum_{\frac{d+\sqrt{d}}{2}+1}^{d}\abs{v_x}^2\sim\Gamma\left(\frac{d-\sqrt{d}}{2},1\right)$, thus $\sum_{x=1}^{\frac{d+\sqrt{d}}{2}}\abs{v_x}^2\Big/\sum_{\frac{d+\sqrt{d}}{2}+1}^{d}\abs{v_x}^2\sim\beta'\left(\frac{d+\sqrt{d}}{2},\frac{d-\sqrt{d}}{2}\right)$, which is the standard beta prime distribution with parameters $\frac{d+\sqrt{d}}{2}$ and $\frac{d-\sqrt{d}}{2}$. Therefore the random variable $\sum_{x=1}^{\frac{d+\sqrt{d}}{2}}\abs{v_x}^2\Big/\sum_{\frac{d+\sqrt{d}}{2}+1}^{d}\abs{v_x}^2$ has cumulative distribution function 
\begin{equation}
\underset{v\sim\mc{N}_{\mbb{C}}(0,\mbb{I}_d)}{\mathrm{Pr}}\left[\sum_{x=1}^{\frac{d+\sqrt{d}}{2}}\abs{v_x}^2\Bigg/\sum_{\frac{d+\sqrt{d}}{2}+1}^{d}\abs{v_x}^2<y\right]=I_{\frac{y}{1+y}}\left(\frac{d+\sqrt{d}}{2},\frac{d-\sqrt{d}}{2}\right),
\end{equation}
where $I_{\cdot}(\cdot,\cdot)$ is the regularized incomplete beta function. Take $y=1$, we have
\begin{equation}
\underset{\rho\sim\pi_{d,1}^*}{\mathrm{Pr}}\left[\Tr(S\rho)<0\right]=\frac{1}{2}I_{1/2}\left(\frac{d+\sqrt{d}}{2},\frac{d-\sqrt{d}}{2}\right).
\end{equation}
\end{proof}

Now we briefly analyze the complexity of the EDP designed based on $S$.
Firstly, if we randomly select states from $\pi_{d,1}^*$, the average expectation value is exponentially small as
\begin{equation}
	\underset{\psi_{AB}\sim\pi_{d,1}}{\mathbb{E}}\Tr(S\psi_{AB})=\underset{\psi_A,\psi_B\sim\pi_{\sqrt{d},1}}{\mathbb{E}}\Tr[S(\psi_A\otimes\psi_B)]
	=\Tr\left(S\frac{\mathbb{I}_d}{d}\right)=\frac{1}{\sqrt{d}},
\end{equation}
where $\frac{\mathbb{I}_d}{d}$ is the maximally mixed state in $\mathcal{H}_d$.
Then, one can prove that is the state $\rho$ is sampled according to the distribution $\pi_{d,1}$, the variance of the expectation value is also exponentially small,
\begin{equation}
\mathrm{Var}_{\psi_{AB}\sim\pi_{d,1}}\left[\Tr(S\psi_{AB})\right]=\mathbb{E}_{\psi_{AB}\sim\pi_{d,1}}\Tr(S\psi_{AB})^2-\frac{1}{d}.
\end{equation}
Substituting the result in Lemma \ref{lemma:expectation_to_permutation}, we arrive at
\begin{equation}
\mathrm{Var}_{\psi_{AB}\sim\pi_{d,1}}\left[\Tr(S\psi_{AB})\right]=\frac{2d}{d(d+1)}-\frac{1}{d}\sim\frac{1}{d}.
\end{equation}
This means that, when we sample a state $\psi_{AB}$ with negative expectation value $\Tr(S\psi_{AB})<0$, the absolute value of $\Tr(S\psi_{AB})$ normally scales as $\frac{1}{\sqrt{d}}$.
One can prove that with a constant probability, the value of $\Tr(S\psi_{AB})$ is less than $-\frac{1}{\sqrt{d}}$.
Therefore, to satisfy the requirement of soundness, the EDP should have an estimation error approximately $\frac{1}{\sqrt{d}}$.
Combining with the property of SWAP operator mentioned in the main context, this requires the sample complexity of EDP to be approximately $\mathcal{O}(d)$.

\section{Sample Complexity for Pure Bipartite State EDP}\label{sec:bipartite_lower_bound}
\begin{theorem}\label{theorem:exponential_complexity_app}
Given multiple copies of a state $\rho$ that is sampled from $\pi_{d,1}^*$, the minimum number of copies required by a single-copy adaptive EDP with completeness, $\Pr_{\rho\sim\pi_{d,1}^*}\left[\hat{C}(\rho)<0\right]=\frac{1}{4}$, and soundness, $\Pr_{\rho\sim\pi_{d,1}^*}\left[\rho\in\operatorname{ENT}\Big|\hat{C}(\rho)<0\right]=\frac{5}{6}$, scales as $\Theta(d^{1/4})$.
\end{theorem}
\begin{proof}

To prove the lower bound on the sample complexity of EDP, we first find an equivalent state discrimination task and derive the lower bound for that task.
Notice that given an EDP satisfying the requirement in Theorem~\ref{theorem:exponential_complexity_app}, it satisfies
\begin{equation}
\begin{aligned}
\Pr_{\rho\sim\pi_{d,1}^*}\left[\rho\in\operatorname{ENT},\hat{C}(\rho)<0\right]=\Pr_{\rho\sim\pi_{d,1}^*}\left[\rho\in\operatorname{ENT}\Big|\hat{C}(\rho)<0\right]\Pr_{\rho\sim\pi_{d,1}^*}\left[\hat{C}(\rho)<0\right]=\frac{5}{24}
\end{aligned}
\end{equation}
and
\begin{equation}
\begin{aligned}
&\Pr_{\rho\sim\pi_{d,1}^*}\left[\rho\in\operatorname{SEP},\hat{C}(\rho)\ge 0\right]=\Pr_{\rho\sim\pi_{d,1}^*}\left[\rho\in\operatorname{SEP}\right]-\Pr_{\rho\sim\pi_{d,1}^*}\left[\rho\in\operatorname{SEP},\hat{C}(\rho)< 0\right]\\
=&\Pr_{\rho\sim\pi_{d,1}^*}\left[\rho\in\operatorname{SEP}\right]-\Pr_{\rho\sim\pi_{d,1}^*}\left[\hat{C}(\rho)< 0\right]+\Pr_{\rho\sim\pi_{d,1}^*}\left[\rho\in\operatorname{ENT},\hat{C}(\rho)< 0\right]=\frac{11}{24}.
\end{aligned}
\end{equation}
Consider the following two pure state sets with distributions defined as:
\begin{enumerate}
\item $s_1$: $\rho=\ketbra{\psi}{\psi}$, where $\ket{\psi}\in \mathcal{H}_d$ is a Haar random pure state.
\item $s_2$: $\rho=\ketbra{\psi_A}{\psi_A}\otimes\ketbra{\psi_B}{\psi_B}$, where $\ket{\psi_A},\ket{\psi_B}\in \mathcal{H}_{\sqrt{d}}$ are Haar random pure states.
\end{enumerate}
Suppose a referee tosses a fair coin to select a set, and then chooses a state $\rho$ from the selected set according to the distribution shown above. 
Then the referee gives $T$ copies of $\rho$ to us. 
Our task is to decide which set $\rho$ is from, with $T$ copies of $\rho$.
Notice that state sampled from $s_1$ is entangled with probability $1$ and state sampled from $s_2$ is entangled with probability $0$.
Thus, if we can use these $T$ copies of $\rho$ to achieve the EDP with requirements in Theorem~\ref{theorem:exponential_complexity_app}, we can solve this state discrimination task with success probability of
\begin{equation}
\begin{aligned}
\Pr_{\rho\sim\pi_{d,1}^*}\left[\rho\in s_1,\hat{C}(\rho)<0\right]+\Pr_{\rho\sim\pi_{d,1}^*}\left[\rho\in s_2,\hat{C}(\rho)\ge 0\right]=\frac{2}{3}.
\end{aligned}
\end{equation}
So, it suffices to prove that the sample complexity for solving this state discrimination task has a lower bound $\Omega(d^{1/4})$ for all single-copy adaptive strategies.

\begin{figure}
    \centering
    \includegraphics[width=0.5\textwidth]{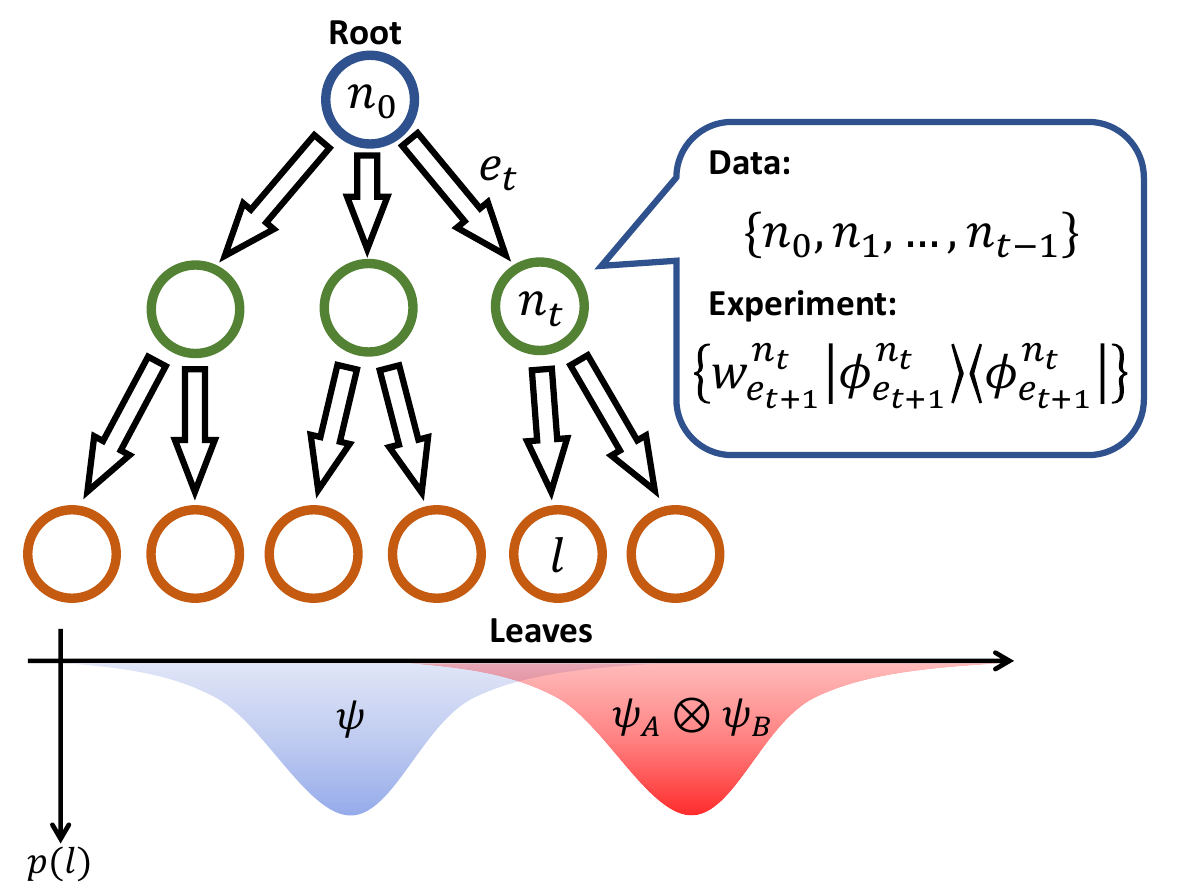}
    \caption{Decision tree for solving the state discrimination task.}
    \label{fig:tree_app}
\end{figure}

For any single-copy adaptive state discrimination protocol and a given state $\rho$, consider the decision tree associated with them, Fig.~\ref{fig:tree_app}. 
We use $l$ to label the leaves of the tree. 
Let $p_{\rho}$ be a probability distribution on the leaves, such that $p_{\rho}(l)$ equals the probability of arriving at the leaf $l$.
As this decision tree covers all the single-copy adaptive protocols to solve the state discrimination tasks, and the success probability depends on the total variation between two probability distributions on leaves, $p_{\psi}(l)$ and $p_{\psi_A\otimes\psi_B}(l)$.
Thus, to find the lower bound of this task, we need to derive the lowest depth that allows the total variation between leaves probability distributions a constant,
\begin{equation}
\operatorname{TV}\left(\underset{\psi\sim \operatorname{Haar}}{\mathbb{E}}p_{\psi},\underset{\psi_A,\psi_B}{\mathbb{E}}p_{\psi_A\otimes\psi_B}\right)\ge\frac{1}{3}.
\end{equation}
To benefit our derivation, we will first calculate the total variations between these two state distributions with the distribution caused by the maximally mixed state $p_{\mbb{I}_d/d}$ and then use the triangle inequality to bound the target total variation.

Since rank-$1$ POVM can simulate arbitrary POVM \cite{Chen2021exponential}, we can assume that all measurements in the decision tree are rank-$1$ POVM.
For a leaf $l$ in the decision tree, there is a unique path from the root node $n_0$ to the leaf node $n_T=l$. 
Denote the nodes on the path by $n_0,\cdots,n_T$ and the edges between them by $e_1,\cdots,e_T$. 
For $t=1,\cdots, T$, denote the rank$-1$ POVM operator associated to $e_{t}$ by $w_{e_{t}}^{n_{t-1}}\ketbra{\phi_{e_{t}}^{n_{t-1}}}{\phi_{e_{t}}^{n_{t-1}}}$, with $w_{e_{t}}^{n_{t-1}}>0$ satisfying the normalization condition $\sum_{e}w_{e}^{n_{t-1}}=d$, where the summation is calculated over all the edges connecting node $n_{t-1}$ with its child nodes. 
To calculate the total variation between the distribution caused by $s_1$ and the maximally mixed state, we calculate the following quantity:

\begin{equation}
\begin{aligned}
\frac{\underset{\psi\sim \operatorname{Haar}}{\mathbb{E}}p_{\psi}(l)}{p_{\mathbb{I}_d/d}(l)}=&\frac{1}{\prod_{t=1}^Tw_{e_t}^{n_{t-1}}\frac{1}{d}}\underset{\psi\sim \operatorname{Haar}}{\mathbb{E}}\prod_{t=1}^{T}w_{e_t}^{n_{t-1}}\braket{\phi_{e_t}^{n_{t-1}}}{\psi}\braket{\psi}{\phi_{e_t}^{n_{t-1}}}\\
=&d^T\underset{\psi\sim \operatorname{Haar}}{\mathbb{E}}\prod_{t=1}^{T}\braket{\phi_{e_t}^{n_{t-1}}}{\psi}\braket{\psi}{\phi_{e_t}^{n_{t-1}}}\\
=&d^T\bra{\Phi}\left(\underset{\psi\sim \operatorname{Haar}}{\mathbb{E}}\ketbra{\psi}{\psi}^{\otimes T}\right)\ket{\Phi}\\
=&\frac{d^T}{\left(d+T-1\right)\cdots\left(d+1\right)d}\bra{\Phi}\left(\sum_{\sigma\in S_T}W_{\sigma}\right)\ket{\Phi},
\end{aligned}
\end{equation}
where $\ket{\Phi}=\otimes_{t=1}^T\ket{\phi_{e_t}^{n_{t-1}}}$, and the last equality is given by Lemma~\ref{lemma:expectation_to_permutation}. With the fact $\bra{\Phi}\left(\sum_{\sigma\in S_T}W_{\sigma}\right)\ket{\Phi}\geq 1$ \cite{Chen2021exponential}, we have
\begin{equation}
\begin{aligned}
\frac{\underset{\psi\sim \operatorname{Haar}}{\mathbb{E}}p_{\psi}(l)}{p_{\mathbb{I}_d/d}(l)}\geq&\frac{d^T}{\left(d+T-1\right)\cdots\left(d+1\right)d}=\frac{1}{\left(1+\frac{T-1}{d}\right)\cdots\left(1+\frac{1}{d}\right)1}\geq\left(1+\frac{T-1}{d}\right)^{-(T-1)}.
\end{aligned}
\end{equation}
We can use it to estimate the total variation as
\begin{equation}\label{eq:TV_Epsi_I}
\begin{aligned}
\operatorname{TV}\left(\underset{\psi\sim \operatorname{Haar}}{\mathbb{E}}p_{\psi},p_{\mathbb{I}_d/d}\right)=\sum_{l}p_{\mathbb{I}_d/d}(l)\max\left\{0,1-\frac{\underset{\psi\sim \operatorname{Haar}}{\mathbb{E}}p_{\psi}(l)}{p_{\mathbb{I}_d/d}(l)}\right\}\leq1-\left(1+\frac{T-1}{d}\right)^{-(T-1)}.
\end{aligned}
\end{equation}

On the other hand, we have
\begin{equation}
\begin{aligned}
\frac{\underset{\psi_A,\psi_B}{\mathbb{E}}p_{\psi_A\otimes\psi_B}(l)}{p_{\mathbb{I}_d/d}(l)}&=\frac{1}{\prod_{t=1}^T\left(w_{e_t}^{n_{t-1}}\frac{1}{d}\right)}\underset{\psi_A,\psi_B}{\mathbb{E}}\prod_{t=1}^{T}w_{e_t}^{n_{t-1}}\braket{\phi_{e_t}^{n_{t-1}}}{\psi_A\psi_B}\braket{\psi_A\psi_B}{\phi_{e_t}^{n_{t-1}}}\\
&=d^T\underset{\psi_A,\psi_B}{\mathbb{E}}\prod_{i=1}^{T}\braket{\phi_{e_t}^{n_{t-1}}}{\psi_A\psi_B}\braket{\psi_A\psi_B}{\phi_{e_t}^{n_{t-1}}}\\
&=d^T\bra{\Phi}\left(\underset{\psi_A,\psi_B}{\mathbb{E}}(\ketbra{\psi_A}{\psi_A}\otimes\ketbra{\psi_B}{\psi_B})^{\otimes T}\right)\ket{\Phi}\\
&=d^T\bra{\Phi}\mc{V}_T\left(\underset{\psi_A}{\mathbb{E}}\ketbra{\psi_A}{\psi_A}^{\otimes T}\otimes\underset{\psi_B}{\mathbb{E}}\ketbra{\psi_B}{\psi_B}^{\otimes T}\right)\mathcal{V}_T^{\dagger}\ket{\Phi}\\
&=\left(\frac{\left(d^{1/2}\right)^T}{\left(d^{1/2}+T-1\right)\cdots\left(d^{1/2}+1\right)d^{1/2}}\right)^2\bra{\Phi}\mathcal{V}_T\left(\sum_{\sigma_1\in S_T}W_{\sigma_1}\right)\otimes\left(\sum_{\sigma_2\in S_T}W_{\sigma_2}\right)\mathcal{V}_T^{\dagger}\ket{\Phi},
\end{aligned}
\end{equation}
where the last equality can be obtained by Lemma~\ref{lemma:expectation_to_permutation}. By Lemma~\ref{lemma:pure_product} below, we have
\begin{equation}
\begin{aligned}
\frac{\underset{\psi_A,\psi_B}{\mathbb{E}}p_{\psi_A\otimes\psi_B}(l)}{p_{\mathbb{I}_d/d}(l)}\geq&\left(\frac{\left(d^{1/2}\right)^T}{\left(d^{1/2}+T-1\right)\cdots\left(d^{1/2}+1\right)d^{1/2}}\right)^2\\
=&\left(\frac{1}{\left(1+\frac{T-1}{d^{1/2}}\right)\cdots\left(1+\frac{1}{d^{1/2}}\right)1}\right)^2\geq\left(1+\frac{T-1}{d^{1/2}}\right)^{-2(T-1)}.
\end{aligned}
\end{equation}
We can estimate the total variation as
\begin{equation}
\begin{aligned}
\operatorname{TV}\left(\underset{\psi_A,\psi_B}{\mathbb{E}}p_{\psi_A\otimes\psi_B},p_{\mathbb{I}_d/d}\right)=&\sum_{l}p_{\mathbb{I}_d/d}(l)\max\left\{0,1-\frac{\underset{\psi_A,\psi_B}{\mathbb{E}}p_{\psi_A\otimes\psi_B}(l)}{p_{\mathbb{I}_d/d}(l)}\right\}\leq1-\left(1+\frac{T-1}{d^{1/2}}\right)^{-2(T-1)}.
\end{aligned}
\end{equation}
With triangle inequality, we have
\begin{equation}\label{eq:TV1_UpperBound}
\begin{aligned}
\operatorname{TV}\left(\underset{\psi\sim \operatorname{Haar}}{\mathbb{E}}p_{\psi},\underset{\psi_A,\psi_B}{\mathbb{E}}p_{\psi_A\otimes\psi_B}\right)\leq&\operatorname{TV}\left(\underset{\psi\sim \operatorname{Haar}}{\mathbb{E}}p_{\psi},p_{\mathbb{I}_d/d}\right)+\operatorname{TV}\left(\underset{\psi_A,\psi_B}{\mathbb{E}}p_{\psi_A\otimes\psi_B},p_{\mathbb{I}_d/d}\right)\\
\leq&2-\left(1+\frac{T-1}{d}\right)^{-(T-1)}-\left(1+\frac{T-1}{d^{1/2}}\right)^{-2(T-1)}.
\end{aligned}
\end{equation}
By Le Cam’s two-point method \cite{Yu1997LeCam,Huang2022quantum}, the success probability for any procedure to distinguish between the two sets is upper bounded by
\begin{equation}
\frac{1}{2}+\frac{1}{2}\operatorname{TV}\left(\underset{\psi\sim \operatorname{Haar}}{\mathbb{E}}p_{\psi},\underset{\psi_A,\psi_B}{\mathbb{E}}p_{\psi_A\otimes\psi_B}\right).
\end{equation}
Combined with Eq.~\eqref{eq:TV1_UpperBound}, to achieve a success probability of at least $2/3$, it must hold that
\begin{equation}
\begin{aligned}
\frac{2}{3}\leq\frac{1}{2}+\frac{1}{2}\operatorname{TV}\left(\underset{\psi\sim \operatorname{Haar}}{\mathbb{E}}p_{\psi},\underset{\psi_A,\psi_B}{\mathbb{E}}p_{\psi_A\otimes\psi_B}\right)\leq&\frac{3}{2}-\frac{1}{2}\left(1+\frac{T-1}{d}\right)^{-(T-1)}-\frac{1}{2}\left(1+\frac{T-1}{d^{1/2}}\right)^{-2(T-1)}\\
\leq&\frac{3}{2}-\left(1+\frac{T-1}{d^{1/2}}\right)^{-2(T-1)},
\end{aligned}
\end{equation}
which implies
\begin{equation}
\left(1+\frac{T-1}{d^{1/2}}\right)^{-2(T-1)}\leq\frac{5}{6}.
\end{equation}
Thus we must have
\begin{equation}
T\geq\sqrt{\frac{1}{2}\log(\frac{6}{5})}d^{1/4}+1=\Omega\left(d^{1/4}\right).
\end{equation}

We leave the proof of the upper bound in Appendix~\ref{sec:upper_bound_proof}.
\end{proof}

\begin{lemma}\label{lemma:pure_product}
For any product pure state $\ket{\Phi}=\ket{\phi_1}\otimes\cdots\otimes\ket{\phi_T}$ in $\mathcal{H}_d^{\otimes T}$, we have 
\begin{equation}
\bra{\Phi}\mathcal{V}_T\left(\sum_{\sigma_1\in S_T}W_{\sigma_1}\right)\otimes\left(\sum_{\sigma_2\in S_T}W_{\sigma_2}\right)\mathcal{V}_T^{\dagger}\ket{\Phi}\geq 1,
\end{equation}
where $W_{\sigma_1}$ and $W_{\sigma_2}$ both act on the space of $\mathcal{H}_{\sqrt{d}}^{\otimes T}$
\end{lemma}
\begin{proof}
We prove by induction on $T$. When $T=1$, we have
\begin{equation}
\bra{\phi_1}\mathcal{V}_1\left(\mathbb{I}_{d^{1/2}}\otimes \mathbb{I}_{d^{1/2}}\right)\mathcal{V}_1^{\dagger}\ket{\phi_1}=\braket{\phi_1}{\phi_1}=1.
\end{equation}
Now suppose the statement holds for $T-1$. Denote $\Pi=\frac{1}{T!}\sum_{\sigma\in S_T}W_{\sigma}$ be the projector to the symmetric subspace in of $\mathcal{H}_{\sqrt{d}}^{\otimes T}$, $\widetilde{\Pi}=\frac{1}{(T-1)!}\sum_{\widetilde{\sigma}\in S_{T-1}}W_{\widetilde{\sigma}}$ be the projector to the symmetric subspace in of $\mathcal{H}_{\sqrt{d}}^{\otimes (T-1)}$. We have $\Pi\left(\mbb{I}\otimes\widetilde{\Pi}\right)=\left(\mbb{I}\otimes\widetilde{\Pi}\right)\Pi=\Pi$ and $\widetilde{\Pi}W_{\widetilde{\sigma}}=W_{\widetilde{\sigma}}\widetilde{\Pi}=\widetilde{\Pi},\forall \widetilde{\sigma}\in S_{T-1}$. Note that
\begin{equation}\label{eq:induction_pure_4terms}
\begin{aligned}
&\bra{\Phi}\mathcal{V}_T\left(\sum_{\sigma_1\in S_T}W_{\sigma_1}\right)\otimes\left(\sum_{\sigma_2\in S_T}W_{\sigma_2}\right)\mathcal{V}_T^{\dagger}\ket{\Phi}\\
=&\bra{\Phi}\mathcal{V}_T\left(\left(\mbb{I}\otimes\widetilde{\Pi}\right)\sum_{\sigma_1\in S_T}W_{\sigma_1}\left(\mbb{I}\otimes\widetilde{\Pi}\right)\right)\otimes\left(\left(\mbb{I}\otimes\widetilde{\Pi}\right)\sum_{\sigma_2\in S_T}W_{\sigma_2}\left(\mbb{I}\otimes\widetilde{\Pi}\right)\right)\mathcal{V}_T^{\dagger}\ket{\Phi}\\
=&\bra{\Phi}\mathcal{V}_T\left(\left(\mbb{I}\otimes\widetilde{\Pi}\right)\sum_{\sigma_1(1)=1}W_{\sigma_1}\left(\mbb{I}\otimes\widetilde{\Pi}\right)\right)\otimes\left(\left(\mbb{I}\otimes\widetilde{\Pi}\right)\sum_{\sigma_2(1)=1}W_{\sigma_2}\left(\mbb{I}\otimes\widetilde{\Pi}\right)\right)\mathcal{V}_T^{\dagger}\ket{\Phi}\\
&+\bra{\Phi}\mathcal{V}_T\left(\left(\mbb{I}\otimes\widetilde{\Pi}\right)\sum_{\sigma_1(1)\neq1}W_{\sigma_1}\left(\mbb{I}\otimes\widetilde{\Pi}\right)\right)\otimes\left(\left(\mbb{I}\otimes\widetilde{\Pi}\right)\sum_{\sigma_2(1)=1}W_{\sigma_2}\left(\mbb{I}\otimes\widetilde{\Pi}\right)\right)\mathcal{V}_T^{\dagger}\ket{\Phi}\\
&+\bra{\Phi}\mathcal{V}_T\left(\left(\mbb{I}\otimes\widetilde{\Pi}\right)\sum_{\sigma_1(1)=1}W_{\sigma_1}\left(\mbb{I}\otimes\widetilde{\Pi}\right)\right)\otimes\left(\left(\mbb{I}\otimes\widetilde{\Pi}\right)\sum_{\sigma_2(1)\neq 1}W_{\sigma_2}\left(\mbb{I}\otimes\widetilde{\Pi}\right)\right)\mathcal{V}_T^{\dagger}\ket{\Phi}\\
&+\bra{\Phi}\mathcal{V}_T\left(\left(\mbb{I}\otimes\widetilde{\Pi}\right)\sum_{\sigma_1(1)\neq1}W_{\sigma_1}\left(\mbb{I}\otimes\widetilde{\Pi}\right)\right)\otimes\left(\left(\mbb{I}\otimes\widetilde{\Pi}\right)\sum_{\sigma_2(1)\neq1}W_{\sigma_2}\left(\mbb{I}\otimes\widetilde{\Pi}\right)\right)\mathcal{V}_T^{\dagger}\ket{\Phi}.
\end{aligned}
\end{equation}
Denote $\ket{\widetilde{\Phi}}=\ket{\phi_2}\otimes\cdots\otimes\ket{\phi_T}$, the first term of the last equation in Eq.~\eqref{eq:induction_pure_4terms} satisfies
\begin{equation}
\begin{aligned}
&\bra{\Phi}\mathcal{V}_T\left(\left(\mbb{I}\otimes\widetilde{\Pi}\right)\sum_{\sigma_1(1)=1}W_{\sigma_1}\left(\mbb{I}\otimes\widetilde{\Pi}\right)\right)\otimes\left(\left(\mbb{I}\otimes\widetilde{\Pi}\right)\sum_{\sigma_2(1)=1}W_{\sigma_2}\left(\mbb{I}\otimes\widetilde{\Pi}\right)\right)\mathcal{V}_T^{\dagger}\ket{\Phi}\\
=&\braket{\phi_1}{\phi_1}\bra{\widetilde{\Phi}}\mathcal{V}_{T-1}\left(\widetilde{\Pi}\sum_{\widetilde{\sigma}_1\in S_{T-1}}W_{\widetilde{\sigma}_1}\widetilde{\Pi}\right)\otimes\left(\widetilde{\Pi}\sum_{\widetilde{\sigma}_2\in S_{T-1}}W_{\widetilde{\sigma}_2}\widetilde{\Pi}\right)\mathcal{V}_{T-1}^{\dagger}\ket{\widetilde{\Phi}}\\
=&\sum_{\widetilde{\sigma}_1,\widetilde{\sigma}_2\in S_{T-1}}\bra{\widetilde{\Phi}}\mathcal{V}_{T-1}\left(\widetilde{\Pi}W_{\widetilde{\sigma}_1}\widetilde{\Pi}\right)\otimes\left(\widetilde{\Pi}W_{\widetilde{\sigma}_2}\widetilde{\Pi}\right)\mathcal{V}_{T-1}^{\dagger}\ket{\widetilde{\Phi}}\\
=&(T-1)!(T-1)!\bra{\widetilde{\Phi}}\mathcal{V}_{T-1}\widetilde{\Pi}\otimes\widetilde{\Pi}\mathcal{V}_{T-1}^{\dagger}\ket{\widetilde{\Phi}}\\
=&\bra{\widetilde{\Phi}}\mathcal{V}_{T-1}\left(\sum_{\widetilde{\sigma}_1\in S_{T-1}}W_{\widetilde{\sigma}_2}\right)\otimes\left(\sum_{\widetilde{\sigma}_1\in S_{T-1}}W_{\widetilde{\sigma}_2}\right)\mathcal{V}_{T-1}^{\dagger}\ket{\widetilde{\Phi}}\geq 1,
\end{aligned}
\end{equation}
where the last inequality is given by the induction hypothesis. The second term satisfies
\begin{equation}
\begin{aligned}
&\bra{\Phi}\mathcal{V}_T\left(\left(\mbb{I}\otimes\widetilde{\Pi}\right)\sum_{\sigma_1(1)\neq1}W_{\sigma_1}\left(\mbb{I}\otimes\widetilde{\Pi}\right)\right)\otimes\left(\left(\mbb{I}\otimes\widetilde{\Pi}\right)\sum_{\sigma_2(1)= 1}W_{\sigma_2}\left(\mbb{I}\otimes\widetilde{\Pi}\right)\right)\mathcal{V}_T^{\dagger}\ket{\Phi}\\
=&\sum_{\sigma_1(1)\neq1}\sum_{\sigma_2(1)=1}\bra{\Phi}\mathcal{V}_T\left(\left(\mbb{I}\otimes\widetilde{\Pi}\right)W_{\sigma_1}\left(\mbb{I}\otimes\widetilde{\Pi}\right)\right)\otimes\left(\left(\mbb{I}\otimes\widetilde{\Pi}\right)W_{\sigma_2}\left(\mbb{I}\otimes\widetilde{\Pi}\right)\right)\mathcal{V}_T^{\dagger}\ket{\Phi}.
\end{aligned}
\end{equation}
We have
\begin{equation}\label{eq:graph_1}
\bra{\Phi}\mathcal{V}_T\left(\left(\mbb{I}\otimes\widetilde{\Pi}\right)W_{\sigma_1}\left(\mbb{I}\otimes\widetilde{\Pi}\right)\right)\otimes\left(\left(\mbb{I}\otimes\widetilde{\Pi}\right)W_{\sigma_2}\left(\mbb{I}\otimes\widetilde{\Pi}\right)\right)\mathcal{V}_T^{\dagger}\ket{\Phi}\geq0    
\end{equation}
for all $\sigma_1$ and $\sigma_2$ satisfying $\sigma_1(1)\neq1$ and $\sigma_2(1)=1$, whose proof is shown graphically in Fig.~\ref{fig:2th-term-1} using tensor network diagrams. 
\begin{figure}[htbp]
\centering
\includegraphics[width=0.98\textwidth]{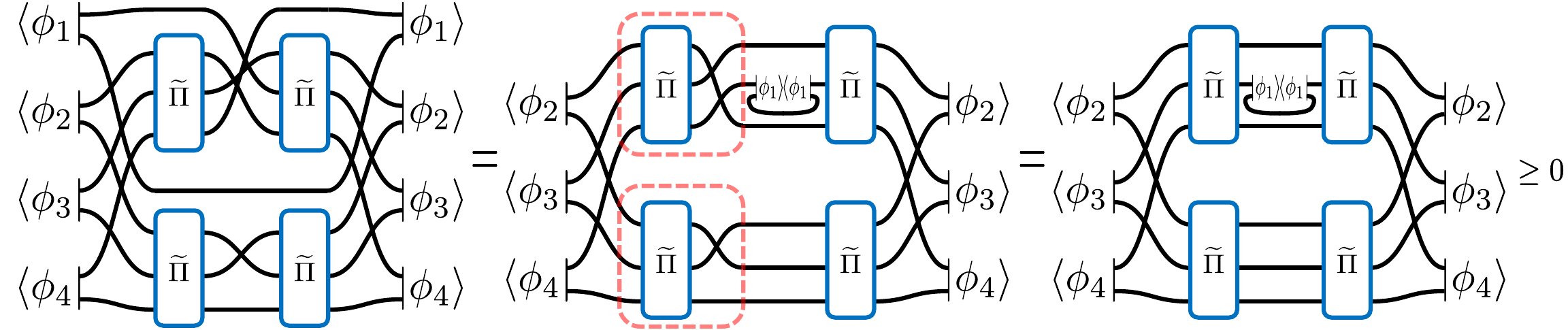}
\caption{We take $T=4$, $W_{\sigma_1}=(1423)$, and $W_{\sigma_2}=(23)$ as an example to illustrate the correctness of Eq.~\eqref{eq:graph_1}. The general situation is similar to this example. The product of two terms in the red dashed box equals $\widetilde{\Pi}$. The last tensor network can be seen as an unnormalized state doing inner product with itself, thus is non-negative.}
\label{fig:2th-term-1}
\end{figure}

Following the same logic, the third term of the last equation in Eq.~\eqref{eq:induction_pure_4terms} is also non-negative. 
The fourth term of the last equation in Eq.~\eqref{eq:induction_pure_4terms} satisfies
\begin{equation}
\begin{aligned}
&\bra{\Phi}\mathcal{V}_T\left(\left(\mbb{I}\otimes\widetilde{\Pi}\right)\sum_{\sigma_1(1)\neq1}W_{\sigma_1}\left(\mbb{I}\otimes\widetilde{\Pi}\right)\right)\otimes\left(\left(\mbb{I}\otimes\widetilde{\Pi}\right)\sum_{\sigma_2(1)\neq1}W_{\sigma_2}\left(\mbb{I}\otimes\widetilde{\Pi}\right)\right)\mathcal{V}_T^{\dagger}\ket{\Phi}\\
=&\sum_{\sigma_1(1)\neq1}\sum_{\sigma_2(1)\neq1}\bra{\Phi}\mathcal{V}_T\left(\left(\mbb{I}\otimes\widetilde{\Pi}\right)W_{\sigma_1}\left(\mbb{I}\otimes\widetilde{\Pi}\right)\right)\otimes\left(\left(\mbb{I}\otimes\widetilde{\Pi}\right)W_{\sigma_2}\left(\mbb{I}\otimes\widetilde{\Pi}\right)\right)\mathcal{V}_T^{\dagger}\ket{\Phi},
\end{aligned}
\end{equation}
where all the terms in the summation are again all non-negative
\begin{equation}\label{eq:graph_2}
\bra{\Phi}\mathcal{V}_T\left(\left(\mbb{I}\otimes\widetilde{\Pi}\right)W_{\sigma_1}\left(\mbb{I}\otimes\widetilde{\Pi}\right)\right)\otimes\left(\left(\mbb{I}\otimes\widetilde{\Pi}\right)W_{\sigma_2}\left(\mbb{I}\otimes\widetilde{\Pi}\right)\right)\mathcal{V}_T^{\dagger}\ket{\Phi}\ge 0.
\end{equation}
We provide graphical proof as shown in Fig.~\ref{fig:4th-term-1}.
\end{proof}

\begin{figure}[htbp]
\centering
\includegraphics[width=0.98\textwidth]{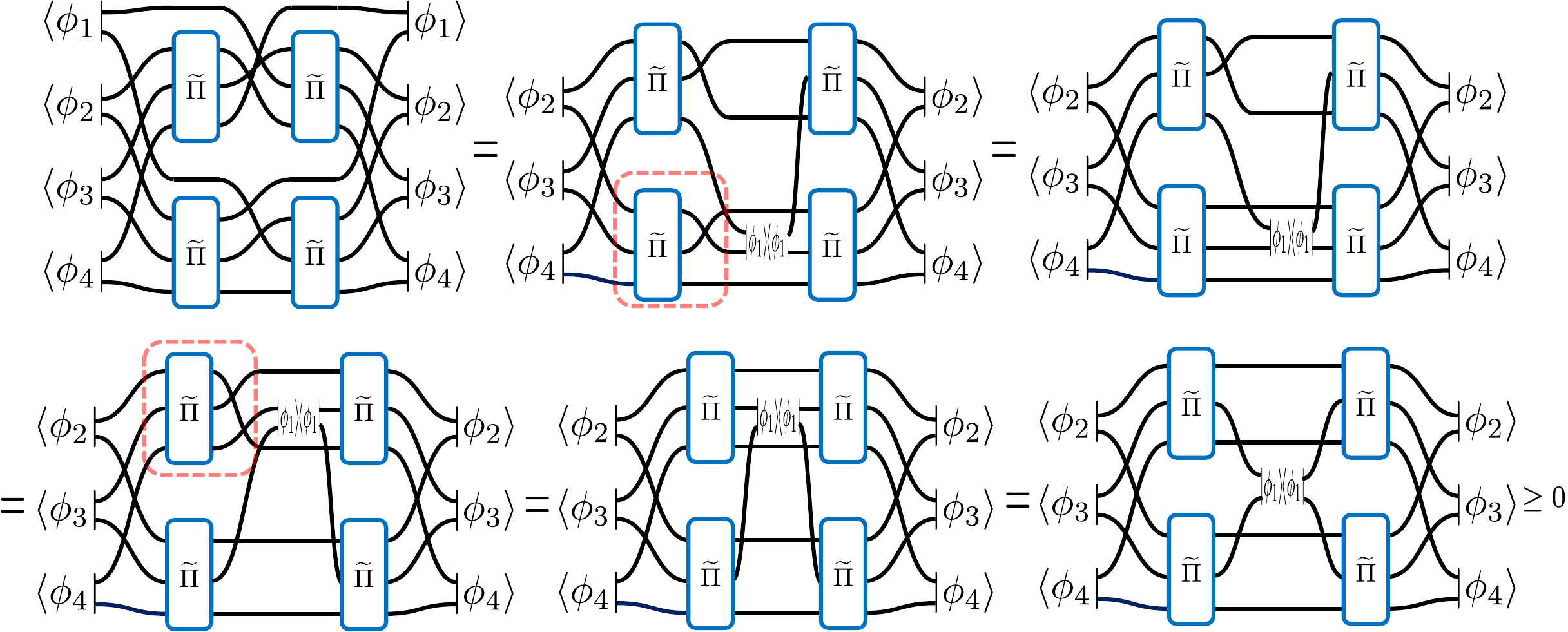}
\caption{We take $T=4$. $W_{\sigma_1}=(1423)$ and $W_{\sigma_2}=(123)$ as an example to illustrate the positivity of Eq.~\eqref{eq:graph_2}. The product of two terms in the red dashed box equals $\widetilde{\Pi}$. The general situation is similar to this example. The last tensor network can be seen as an unnormalized state doing inner product with itself, thus is non-negative.}
\label{fig:4th-term-1}
\end{figure}

\section{Generalization}\label{sec:generalization}
Till now, we mainly focus on the pure bipartite states.
Our result, especially the lower bounds, can be generalized to mixed-state and multipartite cases.
Exponential lower bounds of sample complexities also exist in these cases. 

For mixed-state entanglement detection, we have the following result:

\begin{theorem}\label{theorem:mixed-state}
Given multiple copies of a state $\rho$ that is sampled from $\pi_{d,k}^*$ with $k=\mathcal{O}(d^{3/2})$, all single-copy adaptive EDPs with completeness, $\Pr_{\rho\sim\pi_{d,k}^*}\left[\hat{C}(\rho)<0\right]=\frac{1}{4}$, and soundness, $\Pr_{\rho\sim\pi_{d,k}^*}\left[\rho\in\operatorname{ENT}\Big|\hat{C}(\rho)<0\right]=\frac{5}{6}$, require $\Omega\left((dk)^\frac{1}{4}\right)$ copies of $\rho$.
\end{theorem}

For multipartite cases, we first generalize the definition of $\pi_{d,k}^*$.

\begin{definition}
A state $\rho\in D(\mathcal{H}_d)$ is said to be sampled according to the distribution of $\pi_{d,1}^{(K)*}$ if and only if with probability $0.5$, it is a Haar random pure state $\psi\in\mathcal{H}_{d}$; with probability $0.5$, it is the tensor product of $K$ Haar random pure state $\psi=\bigotimes_{i=1}^K\psi_i$ with $\psi_i\in D(\mathcal{H}_{d^{1/K}})$.
\end{definition}

A $K$-partite state $\rho$ is said to have $K$-partite entanglement iff it cannot be written as
\begin{equation}
\rho = \sum_ip_i\bigotimes_{N=1}^K\rho_{N}^i,
\end{equation}
where $p_i$ is probability satisfying $\sum_ip_i=1$.
Then, the definition of EDP can also be naturally generalized from bipartite to multipartite case. 
We obtain the following result:

\begin{theorem}\label{theorem:multipartite}
Given multiple copies of a state $\rho$ that is sampled from $\pi_{d,k}^{(K)*}$, all single-copy adaptive multipartite EDPs with completeness, $\Pr_{\rho\sim\pi_{d,k}^{(K)*}}\left[\hat{C}(\rho)<0\right]=\frac{1}{4}$, and soundness, $\Pr_{\rho\sim\pi_{d,k}^{(K)*}}\left[\rho\in\operatorname{ENT}\Big|\hat{C}(\rho)<0\right]=\frac{5}{6}$,  require $\Theta\left(d^{\frac{1}{2K}}\right)$ copies of $\rho$.
\end{theorem}

\subsection{Proof of Theorem~\ref{theorem:mixed-state}}

Here we consider the case of $k=\mathcal{O}(d^{3/2})$ as in this range, all state sampled according to $\pi_{d,k}$ is entangled with probability asymptotically one \cite{aubrun2012phase}.
Thus, the analysis of EDP is also equivalent to the analysis of a specific state discrimination task.
Consider the state discrimination task $\{\pi_{d,k}\}$ v.s. $\{\pi_{\sqrt{d},\sqrt{k}}\otimes\pi_{\sqrt{d},\sqrt{k}}\}$. More precisely, consider the following two state sets with distributions defined as:
\begin{enumerate}
\item $\{\rho=\Tr_k(\ketbra{\psi}{\psi})\}$, where $\ket{\psi}\in\mc{H}_{d\times k}$ is a Haar random pure state.
\item $\left\{\rho=\Tr_{\sqrt{k}}(\ketbra{\psi_A}{\psi_A})\otimes\Tr_{\sqrt{k}}(\ketbra{\psi_B}{\psi_B})\right\}$, where $\ket{\psi_A},\ket{\psi_B}\in\mc{H}_{\sqrt{d}\times\sqrt{k}}$ are Haar random pure states.
\end{enumerate}
Suppose a referee tosses a fair coin to select a set, and then chooses a state $\rho$ from the selected set according to the distribution shown above. Then the referee gives $T$ copies of $\rho$ to us. Our task is to decide which set $\rho$ is from, with $T$ copies of $\rho$.

Given an EDP satisfying the requirements in Theorem~\ref{theorem:mixed-state}, we can solve this state discrimination task with a success probability of $2/3$. 
Thus, it suffices to prove:
\begin{proposition}
Using single-copy adaptive measurement, we need at least $T=\Omega\left((dk)^{1/4}\right)$ copies of $\rho$ to find out which set $\rho$ is taken from, with a success probability of at least $2/3$.
\end{proposition}
\begin{proof}
Denote $p_{\psi}$ and $p_{\psi_A,\psi_B}$ be probability distributions on the leaves induced by $\rho=\Tr_2(\ketbra{\psi}{\psi})$ and $\rho=\Tr_2(\ketbra{\psi_A}{\psi_A})\otimes\Tr_2(\ketbra{\psi_B}{\psi_B})$. Consider the following quantity:
\begin{equation}
\begin{aligned}
\frac{\underset{\psi\sim \operatorname{Haar}}{\mathbb{E}}p_{\psi}(l)}{p_{\mathbb{I}_d/d}(l)}=&\frac{1}{\prod_{t=1}^Tw_{e_t}^{n_{t-1}}\frac{1}{d}}\underset{\psi\sim \operatorname{Haar}}{\mathbb{E}}\prod_{t=1}^{T}w_{e_t}^{n_{t-1}}\bra{\phi_{e_t}^{n_{t-1}}}\Tr_k(\ketbra{\psi}{\psi})\ket{\phi_{e_t}^{n_{t-1}}}\\
=&d^T\underset{\psi\sim \operatorname{Haar}}{\mathbb{E}}\prod_{t=1}^{T}\bra{\phi_{e_t}^{n_{t-1}}}\Tr_k(\ketbra{\psi}{\psi})\ket{\phi_{e_t}^{n_{t-1}}}\\
=&d^T\bra{\Phi}\left(\underset{\psi\sim \operatorname{Haar}}{\mathbb{E}}\Tr_k(\ketbra{\psi}{\psi})^{\otimes T}\right)\ket{\Phi}\\
=&d^T\bra{\Phi}\left(\Tr_{\operatorname{even}}\left(\underset{\psi\sim \operatorname{Haar}}{\mathbb{E}}\ketbra{\psi}{\psi}^{\otimes T}\right)\right)\ket{\Phi}\\
=&\frac{d^T}{\left(dk+T-1\right)\cdots\left(dk+1\right)dk}\bra{\Phi}\Tr_{\operatorname{even}}\left(\sum_{\sigma\in S_T}W_{\sigma}\right)\ket{\Phi},
\end{aligned}
\end{equation}
where $W_{\sigma}$ are permutation operators acting on $\mc{H}_{d\times k}^{\otimes T}$, and $\Tr_{\operatorname{even}}$ denotes tracing out systems $2,4,6,\cdots,2T$. By Lemma~\ref{lemma:mixed_psi}, we have
\begin{equation}
\frac{\underset{\psi\sim \operatorname{Haar}}{\mathbb{E}}p_{\psi}(l)}{p_{\mathbb{I}_d/d}(l)}\ge\frac{(dk)^T}{\left(dk+T-1\right)\cdots\left(dk+1\right)dk}=\frac{1}{\left(1+\frac{T-1}{dk}\right)\cdots\left(1+\frac{1}{dk}\right)1}\ge\left(1+\frac{T-1}{dk}\right)^{-(T-1)}.
\end{equation}
On the other hand, we have

\begin{equation}
\begin{aligned}
\frac{\underset{\psi_A,\psi_B}{\mathbb{E}}p_{\psi_A,\psi_B}(l)}{p_{\mathbb{I}_d/d}(l)}&=\frac{1}{\prod_{t=1}^T\left(w_{e_t}^{n_{t-1}}\frac{1}{d}\right)}\underset{\psi_A,\psi_B}{\mathbb{E}}\prod_{t=1}^{T}w_{e_t}^{n_{t-1}}\bra{\phi_{e_t}^{n_{t-1}}}\Tr_{\sqrt{k}}(\ketbra{\psi_A}{\psi_A})\otimes\Tr_{\sqrt{k}}(\ketbra{\psi_B}{\psi_B})\ket{\phi_{e_t}^{n_{t-1}}}\\
&=d^T\underset{\psi_A,\psi_B}{\mathbb{E}}\prod_{i=1}^{T}\bra{\phi_{e_t}^{n_{t-1}}}\Tr_{\sqrt{k}}(\ketbra{\psi_A}{\psi_A})\otimes\Tr_{\sqrt{k}}(\ketbra{\psi_B}{\psi_B})\ket{\phi_{e_t}^{n_{t-1}}}\\
&=d^T\bra{\Phi}\left(\underset{\psi_A,\psi_B}{\mathbb{E}}\left[\Tr_{\sqrt{k}}(\ketbra{\psi_A}{\psi_A})\otimes\Tr_{\sqrt{k}}(\ketbra{\psi_B}{\psi_B})\right]^{\otimes T}\right)\ket{\Phi}\\
&=d^T\bra{\Phi}\mc{V}_T\left(\underset{\psi_A}{\mathbb{E}}\Tr_{\sqrt{k}}(\ketbra{\psi_A}{\psi_A})^{\otimes T}\otimes\underset{\psi_B}{\mathbb{E}}\Tr_{\sqrt{k}}(\ketbra{\psi_B}{\psi_B})^{\otimes T}\right)\mathcal{V}_T^{\dagger}\ket{\Phi}\\
&=d^T\bra{\Phi}\mc{V}_T\left(\Tr_{\operatorname{even}}\left(\underset{\psi_A}{\mathbb{E}}\ketbra{\psi_A}{\psi_A}^{\otimes T}\right)\otimes\Tr_{\operatorname{even}}\left(\underset{\psi_B}{\mathbb{E}}\ketbra{\psi_B}{\psi_B}^{\otimes T}\right)\right)\mathcal{V}_T^{\dagger}\ket{\Phi}\\
&=d^T\left(\frac{1}{\left(\sqrt{dk}+T-1\right)\cdots\left(\sqrt{dk}+1\right)\sqrt{dk}}\right)^2\bra{\Phi}\mathcal{V}_T\Tr_{\operatorname{even}}\left(\sum_{\sigma_1\in S_T}W_{\sigma_1}\right)\otimes\Tr_{\operatorname{even}}\left(\sum_{\sigma_2\in S_T}W_{\sigma_2}\right)\mathcal{V}_T^{\dagger}\ket{\Phi},
\end{aligned}
\end{equation}
where $W_{\sigma_1},W_{\sigma_2}$ are permutation operators acting on $\left(\mc{H}_{\sqrt{dk}}\right)^{\otimes T}$. By Lemma~\ref{lemma:mixed_psiApsiB}, we have
\begin{equation}
\begin{aligned}
\frac{\underset{\psi_A,\psi_B}{\mathbb{E}}p_{\psi_A,\psi_B}(l)}{p_{\mathbb{I}_d/d}(l)}&\ge\left(\frac{\sqrt{dk}^T}{\left(\sqrt{dk}+T-1\right)\cdots\left(\sqrt{dk}+1\right)\sqrt{dk}}\right)^2\\
&=\left(\frac{1}{\left(1+\frac{T-1}{\sqrt{dk}}\right)\cdots\left(1+\frac{1}{\sqrt{dk}}\right)1}\right)^2\ge\left(1+\frac{T-1}{\sqrt{dk}}\right)^{-2(T-1)}.
\end{aligned}
\end{equation}
Putting everything together, we have
\begin{equation}
\begin{aligned}
\operatorname{TV}\left(\underset{\psi\sim \operatorname{Haar}}{\mathbb{E}}p_{\psi},\underset{\psi_A,\psi_B}{\mathbb{E}}p_{\psi_A,\psi_B}\right)\leq&\operatorname{TV}\left(\underset{\psi\sim \operatorname{Haar}}{\mathbb{E}}p_{\psi},p_{\mathbb{I}_d/d}\right)+\operatorname{TV}\left(\underset{\psi_A,\psi_B}{\mathbb{E}}p_{\psi_A,\psi_B},p_{\mathbb{I}_d/d}\right)\\
\leq&2-\left(1+\frac{T-1}{dk}\right)^{-(T-1)}-\left(1+\frac{T-1}{\sqrt{dk}}\right)^{-2(T-1)}\\
\le&2-2\left(1+\frac{T-1}{\sqrt{dk}}\right)^{-2(T-1)}
\end{aligned}
\end{equation}
To achieve a success probability of at least $2/3$, it must hold that
\begin{equation}
\frac{2}{3}\le\frac{1}{2}+\frac{1}{2}\operatorname{TV}\left(\underset{\psi\sim \operatorname{Haar}}{\mathbb{E}}p_{\psi},\underset{\psi_A,\psi_B}{\mathbb{E}}p_{\psi_A,\psi_B}\right)\le\frac{3}{2}-\left(1+\frac{T-1}{\sqrt{dk}}\right)^{-2(T-1)},
\end{equation}
which implies $\left(1+\frac{T-1}{\sqrt{dk}}\right)^{-2(T-1)}\le\frac{5}{6}$, thus $T\ge\Omega\left((dk)^{1/4}\right)$.
\end{proof}

\begin{lemma}\label{lemma:mixed_psi}
For any product state $\ket{\Phi}=\ket{\phi_1}\otimes\cdots\otimes\ket{\phi_T}\in\mathcal{H}_d^{\otimes T}$ we have 
\begin{equation}
\bra{\Phi}\Tr_{\operatorname{even}}\left(\sum_{\sigma\in S_T}W_{\sigma}\right)\ket{\Phi}\geq k^T.
\end{equation}
\end{lemma}
\begin{proof}
We prove by induction on $T$. When $T=1$, we have
\begin{equation}
\bra{\phi_1}\Tr_{k}\left(\mathbb{I}_{d}\otimes\mathbb{I}_{k}\right)\ket{\phi_1}=k\braket{\phi_1}{\phi_1}=k.
\end{equation}
Now suppose the statement holds for $T-1$. Denote $\Pi=\frac{1}{T!}\sum_{\sigma\in S_T}W_{\sigma}$ be the projector to the symmetric subspace in of $\mathcal{H}_{d\times k}^{\otimes T}$, $\widetilde{\Pi}=\frac{1}{(T-1)!}\sum_{\sigma\in S_{T-1}}W_{\sigma}$ be the projector to the symmetric subspace in of $\mathcal{H}_{d\times k}^{\otimes (T-1)}$. We have
\begin{equation}\label{eq:induction_mixed_2terms}
\begin{aligned}
&\bra{\Phi}\Tr_{\operatorname{even}}\left(\sum_{\sigma\in S_T}W_{\sigma}\right)\ket{\Phi}=\bra{\Phi}\Tr_{\operatorname{even}}\left(\left(\mbb{I}\otimes\widetilde{\Pi}\right)\sum_{\sigma\in S_T}W_{\sigma}\left(\mbb{I}\otimes\widetilde{\Pi}\right)\right)\ket{\Phi}\\
=&\bra{\Phi}\Tr_{\operatorname{even}}\left(\left(\mbb{I}\otimes\widetilde{\Pi}\right)\sum_{\sigma(1)=1}W_{\sigma}\left(\mbb{I}\otimes\widetilde{\Pi}\right)\right)\ket{\Phi}+\bra{\Phi}\Tr_{\operatorname{even}}\left(\left(\mbb{I}\otimes\widetilde{\Pi}\right)\sum_{\sigma(1)\neq1}W_{\sigma}\left(\mbb{I}\otimes\widetilde{\Pi}\right)\right)\ket{\Phi}.
\end{aligned}
\end{equation}
Denote $\ket{\widetilde{\Phi}}=\ket{\phi_2}\otimes\cdots\otimes\ket{\phi_T}$. The first term of the last equation in Eq.~\eqref{eq:induction_mixed_2terms} satisfies
\begin{equation}
\begin{aligned}
&\bra{\Phi}\Tr_{\operatorname{even}}\left(\left(\mbb{I}\otimes\widetilde{\Pi}\right)\sum_{\sigma(1)=1}W_{\sigma}\left(\mbb{I}\otimes\widetilde{\Pi}\right)\right)\ket{\Phi}\\
=&k\braket{\phi_1}{\phi_1}\bra{\widetilde{\Phi}}\Tr_{\operatorname{even}}\left(\widetilde{\Pi}\sum_{\widetilde{\sigma}\in S_{T-1}}W_{\widetilde{\sigma}}\widetilde{\Pi}\right)\ket{\widetilde{\Phi}}\\
=&k\sum_{\widetilde{\sigma}\in S_{T-1}}\bra{\widetilde{\Phi}}\Tr_{\operatorname{even}}\left(\widetilde{\Pi}W_{\widetilde{\sigma}}\widetilde{\Pi}\right)\ket{\widetilde{\Phi}}\\
=&k(T-1)!\bra{\widetilde{\Phi}}\Tr_{\operatorname{even}}\left(\widetilde{\Pi}\right)\ket{\widetilde{\Phi}}\\
=&k\bra{\widetilde{\Phi}}\Tr_{\operatorname{even}}\left(\sum_{\widetilde{\sigma}\in S_{T-1}}W_{\widetilde{\sigma}}\right)\ket{\widetilde{\Phi}}\ge k^T,
\end{aligned}
\end{equation}
where the last inequality is given by the induction hypothesis. The second term of the last equation in Eq.~\eqref{eq:induction_mixed_2terms} satisfies
\begin{equation}
\begin{aligned}
\bra{\Phi}\Tr_{\operatorname{even}}\left(\left(\mbb{I}\otimes\widetilde{\Pi}\right)\sum_{\sigma(1)\neq1}W_{\sigma}\left(\mbb{I}\otimes\widetilde{\Pi}\right)\right)\ket{\Phi}=\sum_{\sigma(1)\neq1}\bra{\Phi}\Tr_{\operatorname{even}}\left(\left(\mbb{I}\otimes\widetilde{\Pi}\right)W_{\sigma}\left(\mbb{I}\otimes\widetilde{\Pi}\right)\right)\ket{\Phi}.
\end{aligned}
\end{equation}
We have
\begin{equation}\label{eq:graph_3}
\bra{\Phi}\Tr_{\operatorname{even}}\left(\left(\mbb{I}\otimes\widetilde{\Pi}\right)W_{\sigma}\left(\mbb{I}\otimes\widetilde{\Pi}\right)\right)\ket{\Phi}\geq0
\end{equation}
for all $W_{\sigma}$ satisfying $\sigma(1)\neq1$, which has a graphical proof as shown in Fig.~\ref{fig:graph_3}.

\end{proof}

\begin{figure}[t]
\centering
\includegraphics[width=0.95\textwidth]{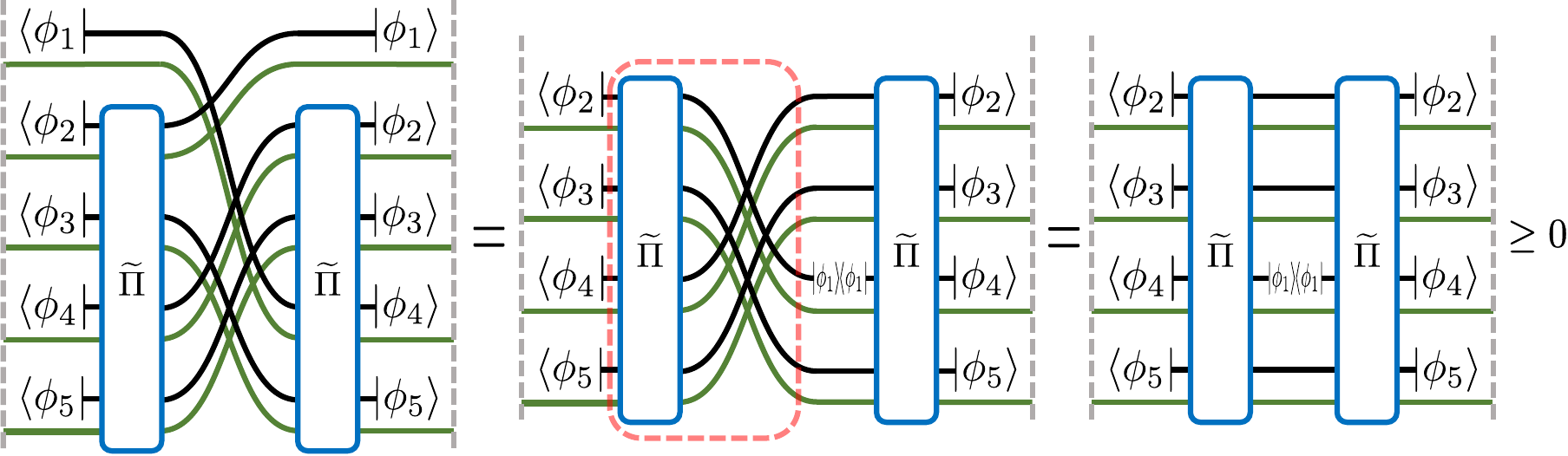}
\caption{We take $W_{\sigma}=(124)(35)$ as an example to illustrate the correctness of Eq.~\eqref{eq:graph_3}. The general situation is similar to this example. We use black and green legs to represent the indices of subsystems $\mc{H}_d$ and $\mc{H}_k$, respectively. The gray dashed lines denote the trace function, which is represented by connecting legs in tensor network diagrams. The product of two terms in the red dashed box equals $\widetilde{\Pi}$. The last tensor network can be seen as an unnormalized state doing inner product with itself, thus is non-negative.}
\label{fig:graph_3}
\end{figure}

\begin{lemma}\label{lemma:mixed_psiApsiB}
For any product pure state $\ket{\Phi}=\ket{\phi_1}\otimes\cdots\otimes\ket{\phi_T}$ in $\mathcal{H}_d^{\otimes T}$, we have
\begin{equation}
\bra{\Phi}\mathcal{V}_T\Tr_{\operatorname{even}}\left(\sum_{\sigma_1\in S_T}W_{\sigma_1}\right)\otimes\Tr_{\operatorname{even}}\left(\sum_{\sigma_2\in S_T}W_{\sigma_2}\right)\mathcal{V}_T^{\dagger}\ket{\Phi}\geq k^T.
\end{equation}
\end{lemma}
\begin{proof}
We prove by induction on $T$. When $T=1$, we have
\begin{equation}
\bra{\phi_1}\Tr_{2,4}\left(\mathbb{I}_{\sqrt{d}}\otimes\mathbb{I}_{\sqrt{k}}\otimes\mathbb{I}_{\sqrt{d}}\otimes\mathbb{I}_{\sqrt{k}}\right)\ket{\phi_1}=k\braket{\phi_1}{\phi_1}=k.
\end{equation}
Now suppose the statement holds for $T-1$. We have
\begin{equation}\label{eq:induction_mixed_4terms}
\begin{aligned}
&\bra{\Phi}\mathcal{V}_T\Tr_{\operatorname{even}}\left(\sum_{\sigma_1\in S_T}W_{\sigma_1}\right)\otimes\Tr_{\operatorname{even}}\left(\sum_{\sigma_2\in S_T}W_{\sigma_2}\right)\mathcal{V}_T^{\dagger}\ket{\Phi}\\
=&\bra{\Phi}\mathcal{V}_T\Tr_{\operatorname{even}}\left(\left(\mbb{I}\otimes\widetilde{\Pi}\right)\sum_{\sigma_1\in S_T}W_{\sigma_1}\left(\mbb{I}\otimes\widetilde{\Pi}\right)\right)\otimes\Tr_{\operatorname{even}}\left(\left(\mbb{I}\otimes\widetilde{\Pi}\right)\sum_{\sigma_2\in S_T}W_{\sigma_2}\left(\mbb{I}\otimes\widetilde{\Pi}\right)\right)\mathcal{V}_T^{\dagger}\ket{\Phi}\\
=&\bra{\Phi}\mathcal{V}_T\Tr_{\operatorname{even}}\left(\left(\mbb{I}\otimes\widetilde{\Pi}\right)\sum_{\sigma_1(1)=1}W_{\sigma_1}\left(\mbb{I}\otimes\widetilde{\Pi}\right)\right)\otimes\Tr_{\operatorname{even}}\left(\left(\mbb{I}\otimes\widetilde{\Pi}\right)\sum_{\sigma_2(1)=1}W_{\sigma_2}\left(\mbb{I}\otimes\widetilde{\Pi}\right)\right)\mathcal{V}_T^{\dagger}\ket{\Phi}\\
+&\bra{\Phi}\mathcal{V}_T\Tr_{\operatorname{even}}\left(\left(\mbb{I}\otimes\widetilde{\Pi}\right)\sum_{\sigma_1(1)\neq1}W_{\sigma_1}\left(\mbb{I}\otimes\widetilde{\Pi}\right)\right)\otimes\Tr_{\operatorname{even}}\left(\left(\mbb{I}\otimes\widetilde{\Pi}\right)\sum_{\sigma_2(1)=1}W_{\sigma_2}\left(\mbb{I}\otimes\widetilde{\Pi}\right)\right)\mathcal{V}_T^{\dagger}\ket{\Phi}\\
+&\bra{\Phi}\mathcal{V}_T\Tr_{\operatorname{even}}\left(\left(\mbb{I}\otimes\widetilde{\Pi}\right)\sum_{\sigma_1(1)=1}W_{\sigma_1}\left(\mbb{I}\otimes\widetilde{\Pi}\right)\right)\otimes\Tr_{\operatorname{even}}\left(\left(\mbb{I}\otimes\widetilde{\Pi}\right)\sum_{\sigma_2(1)\neq 1}W_{\sigma_2}\left(\mbb{I}\otimes\widetilde{\Pi}\right)\right)\mathcal{V}_T^{\dagger}\ket{\Phi}\\
+&\bra{\Phi}\mathcal{V}_T\Tr_{\operatorname{even}}\left(\left(\mbb{I}\otimes\widetilde{\Pi}\right)\sum_{\sigma_1(1)\neq1}W_{\sigma_1}\left(\mbb{I}\otimes\widetilde{\Pi}\right)\right)\otimes\Tr_{\operatorname{even}}\left(\left(\mbb{I}\otimes\widetilde{\Pi}\right)\sum_{\sigma_2(1)\neq1}W_{\sigma_2}\left(\mbb{I}\otimes\widetilde{\Pi}\right)\right)\mathcal{V}_T^{\dagger}\ket{\Phi}.
\end{aligned}
\end{equation}
The first term of the last equation in Eq.~\eqref{eq:induction_mixed_4terms} satisfies
\begin{equation}
\begin{aligned}
&\bra{\Phi}\mathcal{V}_T\Tr_{\operatorname{even}}\left(\left(\mbb{I}\otimes\widetilde{\Pi}\right)\sum_{\sigma_1(1)=1}W_{\sigma_1}\left(\mbb{I}\otimes\widetilde{\Pi}\right)\right)\otimes\Tr_{\operatorname{even}}\left(\left(\mbb{I}\otimes\widetilde{\Pi}\right)\sum_{\sigma_2(1)=1}W_{\sigma_2}\left(\mbb{I}\otimes\widetilde{\Pi}\right)\right)\mathcal{V}_T^{\dagger}\ket{\Phi}\\
=&k\braket{\phi_1}{\phi_1}\bra{\widetilde{\Phi}}\mathcal{V}_{T-1}\Tr_{\operatorname{even}}\left(\widetilde{\Pi}\sum_{\widetilde{\sigma}_1\in S_{T-1}}W_{\widetilde{\sigma}_1}\widetilde{\Pi}\right)\otimes\Tr_{\operatorname{even}}\left(\widetilde{\Pi}\sum_{\widetilde{\sigma}_2\in S_{T-1}}W_{\widetilde{\sigma}_2}\widetilde{\Pi}\right)\mathcal{V}_{T-1}^{\dagger}\ket{\widetilde{\Phi}}\\
=&k\sum_{\widetilde{\sigma}_1,\widetilde{\sigma}_2\in S_{T-1}}\bra{\widetilde{\Phi}}\mathcal{V}_{T-1}\Tr_{\operatorname{even}}\left(\widetilde{\Pi}W_{\widetilde{\sigma}_1}\widetilde{\Pi}\right)\otimes\Tr_{\operatorname{even}}\left(\widetilde{\Pi}W_{\widetilde{\sigma}_2}\widetilde{\Pi}\right)\mathcal{V}_{T-1}^{\dagger}\ket{\widetilde{\Phi}}\\
=&k(T-1)!(T-1)!\bra{\widetilde{\Phi}}\mathcal{V}_{T-1}\Tr_{\operatorname{even}}(\widetilde{\Pi})\otimes\Tr_{\operatorname{even}}(\widetilde{\Pi})\mathcal{V}_{T-1}^{\dagger}\ket{\widetilde{\Phi}}\\
=&k\bra{\widetilde{\Phi}}\mathcal{V}_{T-1}\Tr_{\operatorname{even}}\left(\sum_{\widetilde{\sigma}_1\in S_{T-1}}W_{\widetilde{\sigma}_1}\right)\otimes\Tr_{\operatorname{even}}\left(\sum_{\widetilde{\sigma}_2\in S_{T-1}}W_{\widetilde{\sigma}_2}\right)\mathcal{V}_{T-1}^{\dagger}\ket{\widetilde{\Phi}}\geq k^T,
\end{aligned}
\end{equation}
where $\ket{\widetilde{\Phi}}=\ket{\phi_2}\otimes\cdots\otimes\ket{\phi_T}$ and the last inequality is given by the induction hypothesis. 

Using proof techniques in Lemma~\ref{lemma:pure_product} and Lemma~\ref{lemma:mixed_psi}, we can show that the summands in the second, third, and fourth terms of the last equation in Eq.~\eqref{eq:induction_mixed_4terms} are all non-negative.
\end{proof}

\subsection{Proof of Theorem~\ref{theorem:multipartite}}
For $K\ge2$, consider the following two pure state sets with distributions defined as:
\begin{enumerate}
\item $s_1$: $\rho=\ketbra{\psi}{\psi}$, where $\ket{\psi}\in \mathcal{H}_d$ is a Haar random pure state.
\item $s_2$: $\rho=\ketbra{\psi_1}{\psi_1}\otimes\cdots\otimes\ketbra{\psi_K}{\psi_K}$, where $\ket{\psi_1},\cdots,\ket{\psi_K}\in \mathcal{H}_{d^{1/K}}$ are Haar random pure states.
\end{enumerate}
Suppose a referee tosses a fair coin to select a set, and then chooses a state $\rho$ from the selected set according to the distribution shown above. Then the referee gives $T$ copies of $\rho$ to us. Our task is to decide which set $\rho$ is from, with $T$ copies of $\rho$.

It is easy to find that the state sampled in $s_1$ is $K$-partite entangled with probability $1$ and $0$ for $s_2$.
Thus, given an EDP satisfying the requirements in Theorem~\ref{theorem:multipartite}, we can solve this state discrimination task with a success probability of $2/3$.
It suffices to prove:

\begin{proposition}
Using single-copy adaptive measurement, we need at least $T=\Theta\left(d^{\frac{1}{2K}}\right)$ copies of $\rho$ to find out which set $\rho$ is taken from, to achieve a success probability of at least $2/3$.
\end{proposition}
\begin{proof}
\begin{equation}
\begin{aligned}
\frac{\underset{\psi_1,\cdots,\psi_K}{\mathbb{E}}p_{\psi_1\otimes\cdots\otimes\psi_K(l)}}{p_{\mathbb{I}_d/d}(l)}&=\frac{1}{\prod_{t=1}^T\left(w_{e_t}^{n_{t-1}}\frac{1}{d}\right)}\underset{\psi_1,\cdots,\psi_K}{\mathbb{E}}\prod_{t=1}^{T}w_{e_t}^{n_{t-1}}\braket{\phi_{e_t}^{n_{t-1}}}{\psi_1\cdots\psi_K}\braket{\psi_1\cdots\psi_K}{\phi_{e_t}^{n_{t-1}}}\\
&=d^T\underset{\psi_1,\cdots,\psi_K}{\mathbb{E}}\prod_{i=1}^{T}\braket{\phi_{e_t}^{n_{t-1}}}{\psi_1\cdots\psi_K}\braket{\psi_1\cdots\psi_K}{\phi_{e_t}^{n_{t-1}}}\\
&=d^T\bra{\Phi}\left(\underset{\psi_1,\cdots,\psi_K}{\mathbb{E}}(\ketbra{\psi_1}{\psi_1}\otimes\cdots\otimes\ketbra{\psi_K}{\psi_K})^{\otimes T}\right)\ket{\Phi}\\
&=d^T\bra{\Phi}\mc{V}^{(K)}_T\left(\underset{\psi_1}{\mathbb{E}}\ketbra{\psi_1}{\psi_1}^{\otimes T}\otimes\cdots\otimes\underset{\psi_K}{\mathbb{E}}\ketbra{\psi_K}{\psi_K}^{\otimes T}\right)\mathcal{V}_T^{(K)\dagger}\ket{\Phi}\\
&=\left(\frac{\left(d^{1/K}\right)^T}{\left(d^{1/K}+T-1\right)\cdots\left(d^{1/K}+1\right)d^{1/K}}\right)^K\bra{\Phi}\mathcal{V}^{(K)}_T\left(\sum_{\sigma_1\in S_T}W_{\sigma_1}\right)\otimes\cdots\otimes\left(\sum_{\sigma_K\in S_T}W_{\sigma_K}\right)\mathcal{V}_T^{(K)\dagger}\ket{\Phi},
\end{aligned}
\end{equation}
Similar to the proof of Lemma~\ref{lemma:pure_product}, we have
\begin{equation}
\bra{\Phi}\mathcal{V}^{(K)}_T\left(\sum_{\sigma_1\in S_T}W_{\sigma_1}\right)\otimes\cdots\otimes\left(\sum_{\sigma_K\in S_T}W_{\sigma_K}\right)\mathcal{V}_T^{(K)\dagger}\ket{\Phi}\ge1,
\end{equation}
which implies
\begin{equation}
\begin{aligned}
\frac{\underset{\psi_1,\cdots,\psi_K}{\mathbb{E}}p_{\psi_1\otimes\cdots\otimes\psi_K(l)}}{p_{\mathbb{I}_d/d}(l)}\geq&\left(\frac{\left(d^{1/K}\right)^T}{\left(d^{1/K}+T-1\right)\cdots\left(d^{1/K}+1\right)d^{1/K}}\right)^K\\
=&\left(\frac{1}{\left(1+\frac{T-1}{d^{1/K}}\right)\cdots\left(1+\frac{1}{d^{1/K}}\right)1}\right)^K\geq\left(1+\frac{T-1}{d^{1/K}}\right)^{-K(T-1)}.
\end{aligned}
\end{equation}
Now we have
\begin{equation}
\begin{aligned}
\operatorname{TV}\left(\underset{\psi_1,\cdots,\psi_K}{\mathbb{E}}p_{\psi_1\otimes\cdots\otimes\psi_K},p_{\mathbb{I}_d/d}\right)=&\sum_{l}p_{\mathbb{I}_d/d}(l)\max\left\{0,1-\frac{\underset{\psi_1,\cdots,\psi_K}{\mathbb{E}}p_{\psi_1\otimes\cdots\otimes\psi_K(l)}}{p_{\mathbb{I}_d/d}(l)}\right\}\leq1-\left(1+\frac{T-1}{d^{1/K}}\right)^{-K(T-1)}.
\end{aligned}
\end{equation}
By triangle inequality and Eq.~\eqref{eq:TV_Epsi_I}, we have
\begin{equation}
\begin{aligned}
\operatorname{TV}\left(\underset{\psi\sim \operatorname{Haar}}{\mathbb{E}}p_{\psi},\underset{\psi_1,\cdots,\psi_K}{\mathbb{E}}p_{\psi_1\otimes\cdots\otimes\psi_K}\right)\leq&\operatorname{TV}\left(\underset{\psi\sim \operatorname{Haar}}{\mathbb{E}}p_{\psi},p_{\mathbb{I}_d/d}\right)+\operatorname{TV}\left(\underset{\psi_1,\cdots,\psi_K}{\mathbb{E}}p_{\psi_1\otimes\cdots\otimes\psi_K},p_{\mathbb{I}_d/d}\right)\\
\leq&2-\left(1+\frac{T-1}{d}\right)^{-(T-1)}-\left(1+\frac{T-1}{d^{1/K}}\right)^{-K(T-1)}\\
\le&2-2\left(1+\frac{T-1}{d^{1/K}}\right)^{-K(T-1)}.
\end{aligned}
\end{equation}
To achieve a success probability of at least $2/3$, it must hold that
\begin{equation}
\begin{aligned}
\frac{2}{3}\leq\frac{1}{2}+\frac{1}{2}\operatorname{TV}\left(\underset{\psi\sim \operatorname{Haar}}{\mathbb{E}}p_{\psi},\underset{\psi_1,\cdots,\psi_K}{\mathbb{E}}p_{\psi_1\otimes\cdots\otimes\psi_K}\right)\leq\frac{3}{2}-\left(1+\frac{T-1}{d^{1/K}}\right)^{-K(T-1)},
\end{aligned}
\end{equation}
which implies
$\left(1+\frac{T-1}{d^{1/K}}\right)^{-K(T-1)}\leq\frac{5}{6}$,
thus we must have
\begin{equation}
T\geq\sqrt{\frac{1}{K}\log(\frac{6}{5})}d^{\frac{1}{2K}}+1=\Omega\left(d^{\frac{1}{2K}}\right).
\end{equation}
We also leave the proof of the upper bound in Appendix \ref{sec:upper_bound_proof}.
\end{proof}

\section{Sample Complexity Upper Bound}\label{sec:upper_bound_proof}
To derive the sample complexity upper bound of the pure state discrimination task stated in Appendix \ref{sec:bipartite_lower_bound}, we need to design a protocol and analyze its complexity. The protocol we choose is to measure the purity of subsystem $A$ using randomized measurements. This protocol works because for pure product state $\psi_A\otimes\psi_B$, the purity of subsystem $A$ is always $\Tr(\rho_A^2)=1$. While for global pure state $\psi_{AB}$, the purity can be exponentially small, 
\begin{equation}
\mathbb{E}_{\psi_{AB}}\left[\Tr(\rho_A^2)\right]=\frac{\sqrt{d}+1}{d+1}.
\end{equation}
Therefore, the purity estimation protocol needs to measure purity with constant accuracy to solve this state discrimination task.

The randomized measurement protocol works as follows. 
\begin{enumerate}
\item Pick $N_U$ different unitaries $U_A$ from the Haar measure distribution.
\item Act each unitary on subsystem $A$ and discard subsystem $B$ to get $U_A\rho_AU_A^\dagger$.
\item Measure each $U_A\rho_AU_A^\dagger$ in computational basis for $N_M$ times to get the data $\{b_{U_A}^i\}_{i=1}^{N_M}$.
\end{enumerate}
After acquiring these measurement data, we perform the following calculation to derive the purity of $\rho_A$,
\begin{equation}
\hat{P}_A=\frac{1}{N_U}\frac{1}{N_M(N_M-1)}\sum_{U_A}\sum_{i\neq j}X(b_{U_A}^i,b_{U_A}^j),
\end{equation}
where $X(b,b^\prime)=-(-d_A)^{\delta_{b,b^\prime}}$.
The unbiasedness of this estimator can be verified using the second-order integral over Haar measure unitaries,
\begin{equation}
\begin{aligned}
\mathbb{E}_{U_A\sim\mathrm{Haar}}\mathbb{E}_{\{b_{U_A}^i\}_i}\left(\hat{P}_A\right)=&\mathbb{E}_{U_A\sim\mathrm{Haar}}\sum_{b,b^\prime}X(b,b^\prime)\mathrm{Pr}(b|U_A)\mathrm{Pr}(b^\prime|U_A)\\
=&\mathbb{E}_{U_A\sim\mathrm{Haar}}\sum_{b,b^\prime}X(b,b^\prime)\bra{b}U_A\rho_AU_A^\dagger\ket{b}\bra{b^\prime}U_A\rho_AU_A^\dagger\ket{b^\prime}\\
=&\mathbb{E}_{U_A\sim\mathrm{Haar}}\Tr\left[X(U_A\rho_AU_A^\dagger)^{\otimes 2}\right]\\
=&\mathbb{E}_{U_A\sim\mathrm{Haar}}\Tr\left[U_A^{\dagger\otimes 2}X U_A^{\otimes 2}\rho_A^{\otimes2}\right]\\
=&\Tr(S_A\rho_A^{\otimes 2})=\Tr(\rho_A^2),
\end{aligned}
\end{equation}
where $X=\sum_{b,b^\prime}X(b,b^\prime)\ketbra{bb^\prime}{bb^\prime}$ is a diagonal matrix satisfying $\mathbb{E}_{U\sim\mathrm{Haar}}U^{\otimes 2}XU^{\dagger \otimes 2}=S$ with $S$ being the SWAP operator.

To derive the sample complexity of this protocol, we need to calculate the variance of this estimator. By definition,
\begin{equation}
\mathrm{Var}(\hat{P}_A)=\mathbb{E}(\hat{P}_A^2)-\mathbb{E}(\hat{P}_A)^2.
\end{equation}
As experiments and the estimators for different choices of unitary $U_A$ are independent, the variance can be simplified into
\begin{equation}
\mathrm{Var}(\hat{P}_A)=\frac{1}{N_U}\mathbb{E}_{U_A\sim\mathrm{Haar}}\mathbb{E}_{\{b_{U_A}^i\}_i}\left[\left(\frac{1}{N_M(N_M-1)}\sum_{i \neq j}X(b_{U_A}^i,b_{U_A}^j)\right)^2-\Tr(\rho_A^2)^2\right].
\end{equation}
By expanding the summation term, we have 
\begin{equation}
\begin{aligned}
\left(\frac{1}{N_M(N_M-1)}\sum_{i\neq j}X(b_{U_A}^i,b_{U_A}^j)\right)^2=&\frac{1}{N_M^2(N_M-1)^2}\sum_{i\neq j}X(b_{U_A}^i,b_{U_A}^j)^2\\
+&\frac{1}{N_M^2(N_M-1)^2}\sum_{i\neq j,k\neq i}X(b_{U_A}^i,b_{U_A}^j)X(b_{U_A}^j,b_{U_A}^k)\\
+&\frac{1}{N_M^2(N_M-1)^2}\sum_{i\neq j\neq k\neq l}X(b_{U_A}^i,b_{U_A}^j)X(b_{U_A}^k,b_{U_A}^l).
\end{aligned}
\end{equation}
Defining $X_3=\sum_{b_1,b_2,b_3}=\sum_{b_1,b_2,b_3}X(b_1,b_2)X(b_2,b_3)\ketbra{b_1b_2b_3}{b_1b_2b_3}$, the variance can be written as
\begin{equation}
\begin{aligned}
\mathrm{Var}(\hat{P}_A)=\frac{1}{N_U}\mathbb{E}_{U_A\sim\mathrm{Haar}}\Big{[}&\frac{2}{N_M(N_M-1)}\Tr\left(U_A^{\dagger\otimes 2}X^2 U_A^{\otimes 2}\rho_A^{\otimes 2}\right)+\frac{4(N_M-2)}{N_M(N_M-1)}\Tr\left(U_A^{\dagger\otimes 3}X_3 U_A^{\otimes 3}\rho_A^{\otimes 3}\right)\\
&+\frac{(N_M-2)(N_M-3)}{N_M(N_M-1)}\Tr\left(U_A^{\dagger\otimes 4}X^{\otimes 2} U_A^{\otimes 4}\rho_A^{\otimes 4}\right)-\Tr(\rho_A^2)^2\Big{]}.
\end{aligned}
\end{equation}
Using our conclusion in Ref.~\cite{liu2022permutation}, the expectation terms in this equation are
\begin{equation}
\begin{aligned}
\mathbb{E}_{U_A\sim\mathrm{Haar}}\Tr\left(U_A^{\dagger\otimes 2}X^2 U_A^{\otimes 2}\rho_A^{\otimes 2}\right)=&d_A+(d_A-1)\Tr(\rho_A^2)\\
\mathbb{E}_{U_A\sim\mathrm{Haar}}\Tr\left(U_A^{\dagger\otimes 3}X_3 U_A^{\otimes 3}\rho_A^{\otimes 3}\right)=&-\frac{1}{d_A+2}\left[1+2\Tr(\rho_A^2)\right]+\frac{3(d_A+1)}{d_A+2}\Tr(\rho_A^3)\sim 3\Tr(\rho_A^3)\\
\mathbb{E}_{U_A\sim\mathrm{Haar}}\Tr\left(U_A^{\dagger\otimes 4}X^{\otimes 2} U_A^{\otimes 4}\rho_A^{\otimes 4}\right)\sim&\Tr(\rho_A^2)^2.
\end{aligned}
\end{equation}
Thus, the variance can be bounded by
\begin{equation}
\mathrm{Var}(\hat{P}_A)\le\frac{1}{N_U}\left[\frac{C_1 d_A}{N_M^2}+\frac{C_2\Tr(\rho_A^3)}{N_M}\right],
\end{equation}
where $C_1$ and $C_2$ are constants that are independent of system size. 
Therefore, to make sure the variance is a constant independent of the system size, it is sufficient for $N_U$ to be a constant and $N_M=O(d_A^{1/2})$. 
Thus the total sample complexity scales as $\mathcal{O}(d^{1/4})$.

The conclusion is easy to generalize to multipartite scenarios, where $d_A=d^{1/K}$ and the sample complexity is $\mathcal{O}(d^{1/2K})$.

\end{document}